\documentclass[runningheads]{llncs}

\usepackage{adjustbox}
\usepackage{appendix}
\usepackage{amsmath}
\usepackage{amssymb}
\usepackage[ruled,linesnumbered]{algorithm2e}
\usepackage[altpo,epsilon]{backnaur}
\usepackage{booktabs}
\usepackage{enumerate}
\usepackage{esvect}
\usepackage{graphicx}
\usepackage[hidelinks]{hyperref}
\usepackage{caption}
\usepackage{floatrow}
\usepackage{subcaption}
\usepackage{pgfplots}
\usepackage{semantic}
\usepackage{stmaryrd}
\usepackage{textcomp}

\usepackage{tikz}
\usetikzlibrary{svg.path}

\usepackage{url}
\usepackage{xcolor}
\usepackage{xspace}
\usepackage{zi4}
\usepackage{multicol}
\usepackage{makecell}
\usepackage{paralist}
\usepackage{ifthen}
\usepackage{comment}

\usepackage[capitalise]{cleveref}
\crefname{algorithm}{Alg.}{Algs.}
\crefname{section}{Sec.}{Secs.}
\crefname{definition}{Def.}{Defs.}
\crefname{table}{Tab.}{Tabs.}
\crefname{example}{Ex.}{Exs.}

\definecolor{light-gray}{gray}{0.5}
\definecolor{Black}{rgb}{0.0, 0.0, 0.0}

\usepackage{listings, xcolor}

\definecolor{verylightgray}{rgb}{.97,.97,.97}

\lstdefinelanguage{Solidity}{
	keywords=[1]{anonymous, assembly, assert, balance, break, call, callcode, case, catch, class, constant, constructor, continue, contract, debugger, default, delegatecall, delete, do, else, event, export, external, false, finally, for, function, gas, if, implements, import, in, indexed, instanceof, interface, internal, is, length, library, log0, log1, log2, log3, log4, memory, modifier, new, payable, pragma, private, protected, public, pure, push, require, return, returns, revert, selfdestruct, send, storage, struct, suicide, super, switch, then, this, throw, transfer, true, try, typeof, using, view, while, with, addmod, ecrecover, keccak256, mulmod, ripemd160, sha256, sha3}, 
	keywordstyle=[1]\color{blue}\bfseries,
	keywords=[2]{address, bool, byte, bytes, bytes1, bytes2, bytes3, bytes4, bytes5, bytes6, bytes7, bytes8, bytes9, bytes10, bytes11, bytes12, bytes13, bytes14, bytes15, bytes16, bytes17, bytes18, bytes19, bytes20, bytes21, bytes22, bytes23, bytes24, bytes25, bytes26, bytes27, bytes28, bytes29, bytes30, bytes31, bytes32, enum, int, int8, int16, int24, int32, int40, int48, int56, int64, int72, int80, int88, int96, int104, int112, int120, int128, int136, int144, int152, int160, int168, int176, int184, int192, int200, int208, int216, int224, int232, int240, int248, int256, mapping, string, uint, uint8, uint16, uint24, uint32, uint40, uint48, uint56, uint64, uint72, uint80, uint88, uint96, uint104, uint112, uint120, uint128, uint136, uint144, uint152, uint160, uint168, uint176, uint184, uint192, uint200, uint208, uint216, uint224, uint232, uint240, uint248, uint256, var, void, ether, finney, szabo, wei, days, hours, minutes, seconds, weeks, years},	
	keywordstyle=[2]\color{teal}\bfseries,
	keywords=[3]{block, blockhash, coinbase, difficulty, gaslimit, number, timestamp, msg, data, gas, sender, sig, value, now, tx, gasprice, origin, value},	
	keywordstyle=[3]\color{violet}\bfseries,
	identifierstyle=\color{black},
	sensitive=false,
	comment=[l]{//},
	morecomment=[s]{/*}{*/},
	commentstyle=\color{gray}\ttfamily,
	stringstyle=\color{red}\ttfamily,
	morestring=[b]',
	morestring=[b]"
}

\lstdefinestyle{solidity}{
	language=Solidity,
	extendedchars=true,
    upquote          = true,%
	basicstyle=\tiny\ttfamily,
    columns          = [c]fixed,%
    aboveskip        = 0mm,%
    belowskip        = 2mm,%
    keepspaces       = true,%
    mathescape       = true,%
	showstringspaces=false,
	showspaces=false,
	numbers=left,
	numberstyle=\tiny\color{Black!70},
	numbersep=4pt,
	tabsize=2,
	breaklines=true,
	showtabs=false,
	captionpos=b,
    escapechar=¤,
    moredelim=**[is][{\btHL[fill=Black!20]}]{°}{°},
    xleftmargin=1.3em,%
}

\makeatletter
\newenvironment{btHighlight}[1][]
{\begingroup\tikzset{bt@Highlight@par/.style={#1}}\begin{lrbox}{\@tempboxa}}
{\end{lrbox}\bt@HL@box[bt@Highlight@par]{\@tempboxa}\endgroup}

\newcommand\btHL[1][]{%
  \begin{btHighlight}[#1]\bgroup\aftergroup\bt@HL@endenv%
}
\def\bt@HL@endenv{%
  \end{btHighlight}%
  \egroup
}
\newcommand{\bt@HL@box}[2][]{%
  \tikz[#1]{%
    \pgfpathrectangle{\pgfpoint{0.3pt}{0pt}}{\pgfpoint{\wd #2}{\ht #2}}%
    \pgfusepath{use as bounding box}%
    \node[anchor=base west,fill=Black!20,outer sep=0pt,inner xsep=0.3pt,inner ysep=0pt,minimum height=\ht\strutbox+0.3pt,#1]{\raisebox{0.3pt}{\strut}\strut\usebox{#2}};
  }%
}
\makeatother

\makeatletter
\newif\if@anonymize

\@anonymizefalse 

\if@anonymize
  \newcommand{\smartace}{\textsc{Tool}\xspace}
  \newcommand\anonymize[1]{Link removed for double-blind review.}
\else
   \newcommand{\smartace}{\textsc{SmartACE}\xspace}
  \newcommand\anonymize[1]{\scriptsize#1}
\fi
\makeatother

\usetikzlibrary{automata,calc,trees,positioning,arrows,chains,shapes.geometric,
  decorations.pathreplacing,decorations.pathmorphing,decorations.text,shapes,
  matrix,shapes.symbols,patterns,shadows,arrows.meta}

\pgfplotsset{compat=newest}
\newcommand{\code}[1]{\lstinline[language=Solidity,basicstyle=\small\ttfamily]
  {#1}}
\newcommand{\mathsol}[1]{\text{\code{#1}}}

\newcommand{\trademark}{\textsuperscript{\textregistered}}

\newcommand{\cA}{\mathcal{A}}
\newcommand{\cC}{\mathcal{C}}
\newcommand{\cN}{\mathcal{N}}

\newcommand{\bbD}{\mathbb{D}}
\newcommand{\nat}{\mathbb{N}}

\newcommand{\sqrbrack}[1]{[#1]}
\newcommand{\bbrackfn}[1]{\llbracket #1 \rrbracket}

\renewcommand{\vec}[1]{\mathbf{#1}}
\newcommand{\vu}{\vec{u}}
\newcommand{\vcv}{\vec{v}}
\newcommand{\vx}{\vec{x}}
\newcommand{\vy}{\vec{y}}

\newcommand{\Land}{\bigwedge}

\renewcommand{\implies}{\Rightarrow}
\renewcommand{\gets}{\mathbin{:=}}

\newcommand{\id}{\textsf{id}}
\newcommand{\rec}{\textsf{map}}

\newcommand{\role}{\textsf{role}}
\newcommand{\data}{\textsf{data}}

\newcommand{\client}{\textsf{client}}
\newcommand{\args}{\textsf{arg}}

\newcommand{\control}{\textsf{control}}
\newcommand{\user}{\textsf{user}}

\newcommand{\inputs}{\textsf{action}}

\newcommand{\lts}{\textsf{lts}}
\newcommand{\local}{\textsf{local}}

\newcommand{\pt}{\textsf{pt}}
\newcommand{\explicit}{\mathit{Explicit}}
\newcommand{\implicit}{\mathit{Implicit}}
\newcommand{\transient}{\mathit{Transient}}

\newcommand{\seahorn}{\textsc{SeaHorn}\xspace}
\newcommand{\verx}{\textsc{VerX}\xspace}

\newcommand{\prop}[1]{{\bf Prop.~#1}}

\newcommand{\contractaceurl}{https://github.com/contract-ace}
\newcommand{\seahornblogurl}{http://seahorn.github.io/blog/}

\newcommand{\qspurl}{https://github.com/quantstamp/qsp-staking-protocol}

\newboolean{extended}   
\setboolean{extended}{true}
\newcommand{\appendixcite}[1]
           {\ifthenelse{\boolean{extended}}
                       {\cref{#1}} {the extended version~\cite{arxiv}}}

\hypersetup{breaklinks=true}

\begin{document}

\ifthenelse{\boolean{extended}}
           {\title{Compositional Verification of Smart Contracts Through Communication Abstraction (Extended)}}
           {\title{Compositional Verification of Smart Contracts Through Communication Abstraction\thanks{This work was supported, in part, by Individual Discovery Grants from the Natural Sciences and Engineering Research Council of Canada, and Ripple Fellowship. Jorge~A. Navas was supported by NSF grant 1816936.}}}

\titlerunning{Compositional Verification of Smart Contracts}
\authorrunning{S. Wesley et al.}
%

\definecolor{orcidgreen}{HTML}{A6CE39}
\renewcommand\orcidID[1]{\href{https://orcid.org/#1}{\raisebox{-.15em}{\resizebox{1em}{1em}{%
      \begin{pgfpicture}
        \pgfsetcolor{orcidgreen}
        \pgfpathsvg{M106.301,35.435c0,19.572-15.863,35.436-35.436,35.436S35.43,55.007,35.43,35.435S51.293,0,70.865,0
	S106.301,15.863,106.301,35.435z}
        \pgfusepath{fill}
        \pgfsetcolor{white}
        \pgfpathsvg{M59.321,32.39v13.398h-4.264V16.14h4.264V32.39z}
        \pgfpathsvg{M65.578,16.112h11.572c9.772,0,15.724,7.226,15.724,14.839c0,7.004-4.816,14.838-15.779,14.838H65.578 V16.112z M69.841,41.941h6.562c8.305,0,12.098-5.039,12.098-10.99c0-3.654-2.215-10.99-11.877-10.99h-6.782V41.941z}
        \pgfpathsvg{M57.189,54.758c-1.55,0-2.796-1.245-2.796-2.796c0-1.523,1.246-2.797,2.796-2.797s2.796,1.273,2.796,2.797 C59.985,53.485,58.74,54.758,57.189,54.758z}
        \pgfusepath{fill}
      \end{pgfpicture}%
  }}}}

\author{Scott Wesley\inst{1}\orcidID{0000-0002-6708-2122} \and
        Maria Christakis\inst{2} \and
        Jorge A. Navas\inst{3}\orcidID{0000-0002-0516-1167} \and
        Richard Trefler\inst{1} \and\\
        Valentin W{\"u}stholz\inst{4} \and
        Arie Gurfinkel\inst{1}\orcidID{0000-0002-5964-6792}}

\authorrunning{S. Wesley et al.}

\institute{University of Waterloo, Canada \and
           MPI-SWS, Germany \and
           SRI International, USA \and
           ConsenSys, Germany}

\maketitle              
\begin{abstract}
  Solidity smart contracts are programs that manage up to $2^{160}$ users on a
  blockchain. Verifying a smart contract relative to all users is intractable
  due to state explosion. Existing solutions either restrict the number of users
  to under-approximate behaviour, or rely on manual proofs. In this paper, we
  present \emph{local bundles} that reduce contracts with arbitrarily many users
  to sequential programs with a few \emph{representative} users. Each
  representative user abstracts concrete users that are locally symmetric to
  each other relative to the contract and the property. Our abstraction is
  semi-automated. The representatives depend on communication patterns, and are
  computed via static analysis. A summary for the behaviour of each
  representative is provided manually, but a default summary is often
  sufficient. Once obtained, a local bundle is amenable to sequential static
  analysis. We show that local bundles are relatively complete for parameterized
  safety verification, under moderate assumptions. We implement local bundle
  abstraction in \smartace, and show order-of-magnitude speedups compared to a
  state-of-the-art verifier.
\end{abstract}


\section{Introduction}
\label{Sec:Intro}

Solidity smart contracts are distributed programs that facilitate information
flow between users. Users alternate and execute predefined transactions, that
each terminate within a predetermined number of steps. Each user (and contract)
is assigned a unique, $160$-bit address, that is used by the smart contract to
map the user to that user's data. In theory, smart contracts are finite-state
systems with $2^{160}$ users. However, in practice, the state space of a smart
contract is huge---with at least $2^{2^{160}}$ states to accommodate all users
and their data (conservatively counting one bit per user). In this paper, we
consider the challenge of automatically verifying Solidity smart contracts that
rely on user data.

A naive solution for smart contract verification is to verify the finite-state
system directly. However, verifying systems with at least $2^{2^{160}}$ states
is intractable. The naive solution fails because the state space is exponential
in the number of users. Instead, we infer correctness from a small number of
representative users to ameliorate state explosion. To restrict a contract to
fewer users, we first generalize to a \emph{family} of finite-state systems
parameterized by the number of users. In this way, smart contract verification
is reduced to parameterized verification.

\begin{figure}[t]
    \centering
\begin{lstlisting}[style=solidity,multicols=2]
contract Auction {
  mapping(address => uint) bids; ¤\label{line:vars-start}¤
  address manager; uint leadingBid; bool stopped; ¤\label{line:vars-end}¤
  °uint _sum;°

  constructor(address mgr) public { manager = mgr; }¤\label{line:constructor}¤

  function bid(uint amount) public { ¤\label{line:bid}¤
    require(msg.sender != manager); ¤\label{line:neq1}¤
    require(amount > leadingBid); ¤\label{line:req2}¤
    require(!stopped);
    °_sum = _sum + amount - bids[msg.sender];°
    bids[msg.sender] = amount;
    leadingBid = amount;
  }

  function withdraw() public { ¤\label{line:withdraw}¤
    require(msg.sender != manager); ¤\label{line:neq2}¤
    require(bids[msg.sender] != leadingBid);  ¤\label{line:influence1}¤
    require(!stopped);
    °_sum = _sum + 0 - bids[msg.sender];°
    bids[msg.sender] = 0;  ¤\label{line:influence2}¤
  }

  function stop() public { ¤\label{line:stop}¤
    require(msg.sender == manager); ¤\label{line:eq1}¤
    stopped = true;
  }
}
\end{lstlisting}
    \caption{A smart contract that implements a simple auction.}
    \label{Fig:Auction}
\end{figure}

\begin{figure}[t]
    \centering
\begin{lstlisting}[style=solidity,multicols=2]
Auction _a = new Auction(address(2));
_a.address = address(1);

while (true) {
  // Applies an interference invariant.
  °uint bid = *;° ¤\label{line:harness-instr-start}¤
  °uint maxBid = _a.leadingBid;°
  °require(bid <= maxBid);°
  °require(bid == maxBid || bid + maxBid <= _a.sum);°
  °_a.bids[address(3)] = bid;° ¤\label{line:harness-instr-end}¤
  // Selects a sender.
  msg.sender = *;
  require(msg.sender > address(1)); ¤\label{line:harness-instr-sender-start}¤
  require(msg.sender < address(5));
  °require(msg.sender < address(4));° ¤\label{line:harness-instr-sender}¤
  // Selects a call.
  if      (*) _a.bid(*);
  else if (*) _a.withdraw();
  else if (*) _a.stop();
}
\end{lstlisting}
    \caption{A harness to verify \prop{1} (ignore the highlighted lines)
             and \prop{2}.}
    \label{Fig:Harness}
\end{figure}

For example, consider \code{Auction} in \cref{Fig:Auction} (for now, ignore the
highlighted lines). In \mbox{\code{Auction},} each user starts with a bid of
$0$. Users alternate, and submit increasingly larger bids, until a designated
manager stops the auction. While the auction is not stopped, a non-leading user
may withdraw their bid\footnote{For simplicity of presentation, we do not use
Ether, Ethereum's native currency.}. \code{Auction} satisfies \prop{1}:
``\emph{Once \mbox{\code{stop()}} is called, all bids are immutable}.'' \prop{1}
is satisfied since \mbox{\code{stop()}} sets \code{stopped} to true, no function
sets \code{stopped} to false, and while \code{stopped} is true neither
\mbox{\code{bid()}} nor \mbox{\code{withdraw()}} is enabled. Formally, \prop{1}
is initially true, and remains true due to \prop{1b}: ``\emph{Once
\mbox{\code{stop()}} is called, \code{stopped} remains true}.'' \prop{1} is said
to be inductive relative to its \emph{inductive strengthening} \prop{1b}. A
\emph{Software Model Checker (SMC)} can establish \prop{1} by an exhaustive
search for its inductive strengthening. However, this requires a bound on the
number of addresses, since a search with all $2^{160}$ addresses is intractable.

A bound of at least four addresses is necessary to represent the zero-account
(i.e., a null user that cannot send transactions), the smart contract account,
the manager, and an arbitrary sender. However, once the arbitrary sender submits
a bid, the sender is now the leading bidder, and cannot withdraw its bid. To
enable \mbox{\code{withdraw()},} a fifth user is required. It follows by applying
the results of~\cite{KaiserKroening2010}, that a bound of five addresses is also
sufficient, since users do not read each other's bids, and
adding a sixth user does not enable additional changes to
\mbox{\code{leadingBid}~\cite{KaiserKroening2010}}. The bounded system, known as
a harness, in \cref{Fig:Harness} assigns the zero-account to address 0, the
smart contract account to address 1, the manager to address 2, the arbitrary
senders to addresses 3 and 4, and then executes an unbounded sequence of
arbitrary function calls. Establishing \prop{1} on the harness requires finding
its inductive strengthening. A strengthening such as \prop{1b} (or, in general,
a counterexample violating \prop{1}) can be found by an SMC, directly on the
harness code.

The above bound for \prop{1} also works for checking all control-reachability properties of \code{Auction}. This, for example, follows by applying the results of~\cite{KaiserKroening2010}. That is, \code{Auction} 
has a \emph{Small Model Property (SMP)}~(e.g.,~\cite{KaiserKroening2010,AbdullaHH13}) for such properties. However, not all contracts enjoy an SMP. Consider
\prop{2}: ``\emph{The sum of all active bids is at least \code{leadingBid}}.''
\code{Auction} satisfies \prop{2} since the leading bid is never withdrawn. To
prove \code{Auction} satisfies \prop{2}, we instrument the code to track the
current sum, through the highlighted lines in \cref{Fig:Auction}. With the
addition of \code{\_sum}, \code{Auction} no longer enjoys an SMP. Intuitively,
each user enables new combinations of \code{\_sum} and \code{leadingBid}. As a
proof, assume that there are $N$ users (other than the zero-account, the smart
contract account, and the manager) and let $S_N = 1 + 2 + \cdots + N$. In every
execution with $N$ users, if \code{leadingBid} is $N + 1$, then \code{\_sum} is
less than $S_{N + 1}$, since active bids are unique and $S_{N + 1}$ is the sum
of $N + 1$ bids from $1$ to $N  + 1$. However, in an execution with $N + 1$
users, if the $i$-th user has a bid of $i$, then \code{leadingBid} is $N + 1$
and \code{\_sum} is $S_{N + 1}$. Therefore, increasing $N$ extends the reachable
combinations of \code{\_sum} and \code{leadingBid}. For example, if $N = 2$,
then $S_3 = 1 + 2 + 3 = 6$. If the leading bid is $3$, then the second highest
bid is at most $2$, and, therefore, $\mathsol{\_sum} \le 5 < S_3$. However, when
$N = 3$, if the three active bids are $\{ 1, 2, 3 \}$, then \code{\_sum} is $S_3$.
Therefore, instrumenting \code{Auction} with \code{\_sum} violates the SMP of the original \code{Auction}.

Despite the absence of such an SMP, each function of \code{Auction} interacts
with at most one user per transaction. Each user is classified as either the
zero-account, the smart contract, the manager, or an arbitrary sender. In fact, all
arbitrary senders are indistinguishable with respect to \prop{2}. For example, if
there are exactly three active bids, $\{ 2, 4, 8 \}$, it does not matter which user
placed which bid. The leading bid is $8$ and the sum of all bids is $14$. On the
other hand, if the leading bid is $8$, then each participant of \code{Auction} must
have a bid in the range of $0$~to~$8$. To take advantage of these classes,
rather than analyze \code{Auction} relative to all $2^{160}$ users, it is
sufficient to analyze \code{Auction} relative to a representative user from each
class. In our running example, there must be representatives for the
zero-account, the smart contract account, the manager, and an (arbitrary)
sender. The key idea is that each representative user can correspond to one or
\emph{many} concrete users.

Intuitively, each representative user summarizes the concrete users in its
class. If a representative's class contains a single concrete user, then there
is no difference between the concrete user and the representative user. For
example, the zero-account, the smart contract account, and the manager each
correspond to single concrete users. The addresses of these users, and in turn,
their bids, are known with absolute certainty. On the other hand, there are many
arbitrary senders. Since senders are indistinguishable from each other, the
precise address of the representative sender is unimportant. What matters is
that the representative sender does not share an address with the zero-account,
the smart contract account, nor the manager. However, this means that at the
start of each transaction the location of the representative sender is not
absolute, and, therefore, the sender has a range of possible bids. To account
for this, we introduce a predicate that is true of all initial bids, and holds
inductively across all transactions. We provide this predicate manually, and use
it to over-approximate all possible bids. An obvious predicate for
\code{Auction} is that all bids are at most \mbox{\code{leadingBid},} but this
predicate is not strong enough to prove \prop{2}. For example, the
representative sender could first place a bid of $10$, and then (spuriously)
withdraw a bid of $5$, resulting in a sum of $5$ but a leading bid of $10$. A
stronger predicate, that is adequate to prove \prop{2}, is given by $\theta_U$:
``\emph{Each bid is at most \mbox{\code{leadingBid}.} If a bid is not
\mbox{\code{leadingBid},} then its sum with \code{leadingBid} is at most
\mbox{\code{\_sum}.}}''

Given $\theta_U$, \prop{2} can be verified by an SMC. This requires a new
harness, with representative, rather than concrete, users. The new harness,
\cref{Fig:Harness} (now including the highlighted lines), is similar to the SMP
harness in that the zero-account, the smart contract account, and the manager
account are assigned to addresses 0, 1, and 2, respectively, followed by an
unbounded sequence of arbitrary calls. However, there is now a single sender
that is assigned to address~3 (line~\ref{line:harness-instr-sender}). That is,
the harness uses a fixed configuration of representatives in which the fourth
representative is the sender. Before each function call, the sender's bid is set
to a non-deterministic value that satisfies $\theta_U$
(lines~\ref{line:harness-instr-start}--\ref{line:harness-instr-end}). If the new
harness and \prop{2} are provided to an SMC, the SMC will find an inductive
strengthening such as, ``\emph{The leading bid is at most the sum of all
bids}.''

The harness in \cref{Fig:Harness} differs from existing smart contract
verification techniques in two ways. First, each address in \cref{Fig:Harness}
is an abstraction of one or more concrete users. Second, \code{msg.sender} is
restricted to a finite address space by
lines~\ref{line:harness-instr-sender-start}~to~\ref{line:harness-instr-sender}.
If these lines are removed, then an inductive invariant must constrain
all cells of \code{bids}, to accommodate \code{bids[msg.sender]}. This requires
quantified invariants over arrays that is challenging to automate. By
introducing
lines~\ref{line:harness-instr-sender-start}~to~\ref{line:harness-instr-sender},
a quantifier-free predicate, such as our $\theta_U$, can directly constrain cell
\code{bids[msg.sender]} instead. Adding lines  \ref{line:harness-instr-sender-start}--\ref{line:harness-instr-sender} makes the contract finite state. Thus, its verification problem is decidable and can be 
handled by existing SMCs. However, as illustrated by \prop{2},
the restriction on each user must not exclude feasible
counterexamples. Finding such a restriction is the focus of this paper.

In this paper, we present a new approach to smart contract verification. We
construct finite-state abstractions of parameterized smart contracts, known as
\emph{local bundles}. A local bundle generalizes the harness in
\cref{Fig:Harness}, and is constructed from a set of representatives and their
predicates. When a local bundle and a property are provided to an SMC, there
are three possible outcomes. First, if a predicate does not over-approximate its
representative, a counterexample to the predicate is returned. Second, if the
predicates do not entail the property, then a counterexample to verification is
returned (this counterexample refutes the proof, rather than the property
itself). Finally, if the predicates do entail the property, then an inductive
invariant is returned. As opposed to deductive smart contract solutions, our approach
finds inductive strengthenings
automatically~\cite{HajduJovanovic2019,ZhongCheang2020}. As opposed to other
model checking solutions for smart contracts, our approach is not limited to pre- and
post-conditions~\cite{KalraGoel2018}, and can scale to $2^{160}$
users~\cite{Kolb2020}.

Key theoretical contributions of this paper are to show that verification
with local bundle abstraction is an instance of Parameterized Compositional Model
Checking (PCMC)~\cite{NamjoshiTrefler2016} and the automation of the side-conditions for its applicability. Specifically, \cref{Thm:SafetyCheck}
shows that the local bundle abstraction is a sound proof rule, and a static analysis algorithm (\code{PTGBuilder} in
\cref{sec:communication}) computes representatives  so that the rule is applicable. Key practical contributions are the implementation and the evaluation of the method in a new smart contract verification tool \smartace,
using \seahorn~\cite{GurfinkelKahsai2015} for SMC. \smartace takes as input a contract
and a predicate. Representatives are inferred automatically from the contract, by
analyzing the communication in each transaction. The predicate is then validated
by \seahorn, relative to the representatives. If the predicate is correct, then a local bundle, as in \cref{Fig:Harness}, is returned.

The rest of the paper is structured as follows. \cref{Sec:Background} reviews
parameterized verification. \cref{Sec:SyntaxSemantics} presents MicroSol, a
subset of Solidity with network semantics. \cref{sec:communication} relates user
interactions to representatives. We formalize user interactions as
\emph{Participation Topologies (PTs)}, and define \emph{PT Graphs (PTGs)} to
over-approximate PTs for arbitrarily many users. Intuitively, each PTG
over-approximates the set of representatives. We show that a PTG is computable
for every MicroSol program.
\cref{sec:locality} defines local bundles and proves that our approach is sound.
\cref{sec:eval} evaluates \smartace and shows that it can outperform \verx, a
state-of-the-art verification tool, on all but one \verx benchmark.

\section{Background}
\label{Sec:Background}
In this section, we briefly recall \emph{Parameterized Compositional Model
Checking (PCMC)}~\cite{NamjoshiTrefler2016}. We write $\vu = ( u_0, \ldots,
u_{n-1} )$ for a vector of $n$ elements, and $\vu_i$ for the $i$-th element of
$\vu$. For a natural number $n \in \nat$, we write $[n]$ for $\{ 0, \ldots, n-1
\}$.

\paragraph{Labeled Transition Systems.}
A \emph{labeled transition system (LTS)}, $M$, is a tuple $(S, P, T, s_0)$,
where $S$ is a set of states, $P$ is a set of actions, $T: S \times P \to 2^S$
is a transition relation, and $s_0 \in S$ is an initial state. $M$ is
\emph{deterministic} if $T$ is a function, $T: S \times P \to S$. A (finite)
\emph{trace} of $M$ is an alternating sequence of states and actions, $(s_0,
p_1, s_1, \ldots, p_k, s_k)$, such that $\forall i \in [k] \cdot s_{i+1} \in
T(s_{i}, p_{i+1})$. A state $s$ is \emph{reachable} in $M$ if $s$ is in some
trace $(s_0, p_1, \ldots, s_k)$ of $M$; that is, $\exists i \in [k + 1] \cdot
s_i = s$. A \emph{safety property} for $M$ is a subset of states (or a
predicate\footnote{Abusing notation, we refer to a subset of states $\varphi$ as
a \emph{predicate} and do not distinguish between the syntactic form of
$\varphi$ and the set of states that satisfy it.}) $\varphi \subseteq S$. $M$
satisfies $\varphi$, written $M \models \varphi$, if every reachable state of
$M$ is in $\varphi$.  

Many transition systems are parameterized. For instance, a client-server
application is parameterized by the number of clients, and an array-manipulating
program is parameterized by the number of cells. In both cases, there is a
single \emph{control process} that interacts with many \emph{user
processes}. Such systems are called \emph{synchronized control-user networks
(SCUNs)}~\cite{NamjoshiTrefler2016}. We let $N$ be the number of processes, and
$[N]$ be the process identifiers. We consider SCUNs in which users only
synchronize with the control process and do not execute code on their own.

An SCUN $\cN$ is a tuple $(S_C, S_U, P_I, P_S, T_I, T_S, c_0, u_0)$, where $S_C$
is a set of control states, $S_U$ a set of user states, $P_I$ a set of internal
actions, $P_S$ a set of synchronized actions, $T_I: S_C \times P_I \to S_C$ an
internal transition function, $T_S: S_C \times S_U \times P_S \to S_C \times
S_U$ a synchronized transition function, $c_0 \in S_C$ is the initial control
state, and $u_0 \in S_U$ is the initial user state. The semantics of $\cN$ are
given by a parameterized LTS, $M(N) \gets (S, P, T, s_0)$, where $S \gets S_C
\times \left( S_U \right)^N$, $P \gets P_I \cup \left( P_S \times [N] \right)$,
$s_0 \gets (c_0, u_0, \ldots, u_0)$, and $T: S \times P \rightarrow S$ such that:
\begin{inparaenum}[(1)]
\item if $p \in P_I$, then $T((c, \vu), p) = (T_I(c, p), \vu)$, and
\item if $(p, i) \in P_S \times [N]$, then $T((c, \vu), (p, i)) = ( c', \vu' )$
      where $( c', \vu'_i ) = T_S(c, \vu_i, p)$, and $\forall j \in [N]
      \backslash \{ i \} \cdot \vu'_j = \vu_j$.
\end{inparaenum}

\paragraph{Parameterized Compositional Model Checking (PCMC).}
Parameterized systems have parameterized
properties~\cite{GurfinkelSharon2016,NamjoshiTrefler2016}. A \emph{$k$-universal
safety property}~\cite{GurfinkelSharon2016} is a predicate $\varphi \subseteq
S_C \times (S_U)^k$. A state $( c, \vu )$ satisfies predicate $\varphi$ if
$\forall \{ i_1, \ldots, i_k \} \subseteq [N] \cdot \varphi(c, \vu_{i_1},
\ldots, \vu_{i_k})$. A parameterized system $M(N)$ satisfies predicate $\varphi$
if \mbox{$\forall N \in \nat \cdot M(N) \models \varphi$}. For example, \prop{1}
(\cref{Sec:Intro}) of \code{SimpleAuction} (\cref{Fig:Auction}) is
$1$-universal: ``\emph{For every user $u$, if \code{stop()} has been called,
then $u$ is immutable}.''

Proofs of $k$-universal safety employ compositional reasoning, e.g.,
\cite{AbdullaHH16,GurfinkelSharon2016,NamjoshiTrefler2016,OwickiGries1976}.
Here, we use PCMC~\cite{NamjoshiTrefler2016}. The keys to PCMC are
\emph{uniformity}---the property that finitely many neighbourhoods are
distinguishable---and a \emph{compositional invariant}---a summary of the
reachable states for each equivalence class, that is closed under the actions of
every other equivalence class. For an SCUN, the compositional invariant is
given by two predicates $\theta_C \subseteq S_C$ and $\theta_U \subseteq S_C
\times S_U$ satisfying:
\begin{enumerate}
\item[\textbf{Initialization}] $c_0 \in \theta_C$ and $( c_0, u_0 ) \in
     \theta_U$;
\item[\textbf{Consecution 1}]  If $c \in \theta_C$, $( c, u ) \in \theta_U$, $p
     \in P_S$, and $( c', u' ) \in T_S( c, u, p )$, then $c' \in \theta_C$ and
     $( c', u' ) \in \theta_U$;
\item[\textbf{Consecution 2}]  If $c \in \theta_C$, $( c, u ) \in \theta_U$, $p
     \in P_C$, and $c' = T_I( c, p )$, then $c' \in \theta_C$ and $( c', u ) \in
     \theta_U$;
\item[\textbf{Non-Interference}]  If $c \in \theta_C$, $( c, u ) \in \theta_U$,
     $( c, v ) \in \theta_U$, $u \neq v$, $p \in P_S$, and $( c', u' ) = T_S( c,
     u, p )$, then $( c', v ) \in \theta_C$.
\end{enumerate}
By PCMC~\cite{NamjoshiTrefler2016}, if $\forall c \in \theta_C \cdot \forall \{
(c, u_1), \ldots, (c, u_k) \} \subseteq \theta_U \cdot \varphi(c, u_1, \ldots,
u_k)$, then $M \models \varphi$. This is as an extension of
Owicki-Gries~\cite{OwickiGries1976}, where $\theta_C$ summarizes the acting
process and $\theta_U$ summarizes the interfering process. For this reason, we
call $\theta_C$ the \emph{inductive invariant} and $\theta_U$ the
\emph{interference invariant}.

\section{MicroSol: Syntax and Semantics}
\label{Sec:SyntaxSemantics}

\begin{figure}[t]
  \scriptsize
  \input{diagrams/microsol_grammar}
  \normalsize
  \vspace{-1.2em}
  \caption{The formal grammar of the MicroSol language.}
  \vspace{-1.5em}
  \label{Fig:MicroSol}
\end{figure}

This section provides network semantics for MicroSol, a subset of
Solidity\footnote{\scriptsize \url{https://docs.soliditylang.org/}}. Like Solidity, MicroSol is an
imperative object-oriented language with built-in communication operations. The
syntax of MicroSol is in \cref{Fig:MicroSol}. MicroSol restricts Solidity to
a core subset of communication features. For example, MicroSol does not include
inheritance, cryptographic operations, or mappings between addresses. In our
evaluation (\cref{sec:eval}), we use a superset of MicroSol, called MiniSol (see
\appendixcite{Appendix:MiniSol}), that extends our semantics to a wider set of smart
contracts. Throughout this section, we illustrate MicroSol using \code{Auction}
in \cref{Fig:Auction}.

A MicroSol \emph{smart contract} is similar to a class in object-oriented
programming, and consists of variables, and transactions (i.e., functions) for
users to call. A transaction is a deterministic sequence of operations. Each
smart contract user has a globally unique identifier, known as
an \emph{address}. We view a smart contract as operating in an SCUN: the control
process executes each transaction sequentially, and the user processes are
contract users that communicate with the control process. Users in the SCUN
enter into a transaction through a synchronized action, then the control process
executes the transaction as an internal action, and finally, the users are
updated through synchronized actions. For simplicity of presentation, each
transaction is given as a global transition.

A constructor is a special transaction that is executed once after contract
creation. Calls to \code{new} (i.e., creating new smart contracts) are
restricted to constructors. \code{Auction} in \cref{Fig:Auction} is a smart
contract that defines a constructor (line~\ref{line:constructor}), three other
functions (lines~\ref{line:bid}, \ref{line:withdraw}, and~\ref{line:stop}), and
four state variables (lines~\ref{line:vars-start}--\ref{line:vars-end}).

MicroSol has four types: \emph{address}, \emph{numeric} (including \code{bool}),
\emph{mapping}, and \emph{contract reference}. Address variables prevent
arithmetic operations, and numeric variables cannot cast to address variables.
Mapping and contract-reference variables correspond to dictionaries and object
pointers in other object-oriented languages. Each typed variable is further
classified as either \emph{state}, \emph{input}, or \emph{local}.  We use
\emph{role} and \emph{data} to refer to state variables of address and numeric
types, respectively. Similarly, we use \emph{client} and \emph{argument} to
refer to inputs of address and numeric types, respectively. In \code{Auction} of
\cref{Fig:Auction}, there is $1$ role \mbox{(\code{manager}),} $2$ contract data
\mbox{(\code{leadingBid}} and \mbox{\code{stopped}),} $1$ mapping
\mbox{(\code{bids})}, $1$ client common to all transactions
\mbox{(\code{msg.sender})}, and at most $1$ argument in any transaction
\mbox{(\code{amount}).}

Note that in MicroSol, \emph{user} denotes any user process within a SCUN. A
\emph{client} is defined relative to a transaction, and denotes a user passed as
an input.

\newcommand{\tran}{\mathit{tr}}
\newcommand{\cM}{\mathcal{M}}

\paragraph{Semantics of MicroSol.}
Let $\cC$ be a MicroSol program with a single transaction $\tran$ (see
\appendixcite{Appendix:Lts} for multiple transactions). An $N$-user \emph{bundle} is an
$N$-user network of several (possibly identical) MicroSol programs. The
semantics of a bundle is an LTS, $\lts( \cC, N ) \gets ( S, P, f, s_0 )$, where
$S_C \gets \control( \cC, [N] )$ is the set of control states, $S_U \gets \user(
\cC, [N] )$, is the set of user states, $s_{\bot}$ is the error state, $S
\subseteq \left( S_C \cup \{ s_{\bot} \} \right) \times (S_U)^N$ is the set of LTS states, $P
\gets \inputs( \cC, [N] )$ is the set of actions, $f: S \times P \rightarrow S$
is the \emph{transition function}, and $s_0$ is the initial state. We assume,
without loss of generality, that there is a single control process\footnote{Restrictions
place on \code{new} ensure that the number of MicroSol smart contracts in a bundle is a
static fact. Therefore, all control states are synchronized, and can be combined into a
product machine.}.

Let $\bbD$ be the set of $256$-bit unsigned integers. The state space of a smart
contract is determined by the address space, $\cA$, and the state variables of
$\cC$. In the case of $\lts( \cC, N )$, the address space is fixed to $\cA =
[N]$. Assume that $n$, $m$, and $k$ are the number of roles, data, and mappings
in $\cC$, respectively. State variables are stored by their numeric indices
(i.e., variable~$0$, $1$, etc.). Then, $\control( \cC, \cA ) \subseteq \cA^n
\times \bbD^m$ and $\user( \cC, \cA ) \subseteq \cA \times \bbD^k$. For $c = (
\vx, \vy ) \in \control( \cC, \cA)$, $\role( c, i ) = \vx_i$ is the $i$-th role
and $\data( c, i ) = \vy_i$ is the $i$-th datum. For $u = (z, \vy) \in \user(
\cC, \cA )$, $z$ is the address of $u$, and $\rec( u ) = \vy$ are the mapping
values of $u$.

Similarly, actions are determined by the address space, $\cA$, and the input variables
of $\tran$. Assume that $q$ and $r$ are the number of clients and arguments of
$\tran$, respectively. Then $\inputs( \cC, \cA ) \subseteq \cA^q \times \bbD^r$.
For $p = (\vx, \vy) \in \inputs( \cC, \cA )$, $\client( p, i ) = \vx_i$ is the
$i$-th client in $p$ and $\args( p, i ) = \vy_i$ is the $i$-th argument in $p$.
For a fixed $p$, we write $f_p( s, \vu )$ to denote $f( (s, \vu), p )$.

The initial state of $\lts( \cC, N )$ is $s_0 \gets (c, \vu) \in \control( \cC,
[n] ) \times \user( \cC, [n] )^N$, where $c = ( \vec{0}, \vec{0} )$, $\forall i
\in [N] \cdot \rec(\vu_i) = \vec{0}$, and $\forall i \in [N] \cdot \id( \vu_i )
= i$. That is, all variables are zero-initialized and each user has a unique
address.

An $N$-user transition function is determined by the (usual) semantics of
$\tran$, and a \emph{bijection} from addresses to user indices, $\cM: \cA
\rightarrow [N]$. If $\cM( a ) = i$, then address $a$ belongs to user $\vu_i$.
In the case of $\lts( \cC, N )$, the $i$-th user has address $i$, so $\cM(i) =
i$. We write $f \gets \bbrackfn{\cC}_{\cM}$, and given an action $p$, $f_p$
updates the state variables according to the source code of $\tran$ with respect
to $\cM$. If an \code{assert} fails or an address is outside of $\cA$,
then the error state $s_{\bot}$ is returned. If a \code{require}
fails, then the state is unchanged. Note that $f$ preserves the address of each
user.
 
For example, $\lts(\text{Auction}, 4) = (S, P, f, s_0)$ is the $4$-user bundle
of \code{Auction}. Assume that $(c, \vu)$ is the state reached after evaluating
the constructor. Then $\role( c, 0 ) = 2$, $\data( c, 0 ) = 0$, $\data( c, 1 ) =
0$, and $\forall i \in [4] \cdot \rec( \vu_i )_0 = 0$. That is, the manager is
at address 2, the leading bid is 0, the auction is not stopped, and there are no
active bids. This is because variables are zero-indexed, and \code{stopped} is
the second numeric variable (i.e., at index 1). If the user at address $3$ placed
a bid of $10$, this corresponds to $p \in P$ such that $\client( p, 0 ) = 3$ and
$\args( p, 0 ) = 10$. A complete LTS for this example is in
\appendixcite{Appendix:Lts}.

\paragraph{Limitations of MicroSol.}
MicroSol places two restrictions on Solidity. First, addresses are not numeric.
We argue that this restriction is reasonable, as address manipulation is a form
of pointer manipulation. Second, \code{new} must only appear in constructors. In
our evaluation (\cref{sec:eval}), all calls to \code{new} could be moved
into a constructor with minimal effort. We emphasize that the second restriction
does not preclude the use of abstract interfaces for arbitrary contracts.

\section{Participation Topology}
\label{sec:communication}

The core functionality of any smart contract is communication between users.
Usually, users communicate by reading from and writing to designated mapping
entries. That is, the communication paradigm is shared memory. However, it is
convenient in interaction analysis to re-imagine smart contracts as having
rendezvous synchronization in which users explicitly participate in message
passing. In this section, we formally re-frame smart contracts with explicit
communication by defining a (semantic) participation topology and its
abstractions.

\begin{figure}[t]
    \centering
    \begin{subfigure}{.49\linewidth}
        \centering
        \includegraphics[scale=0.75]{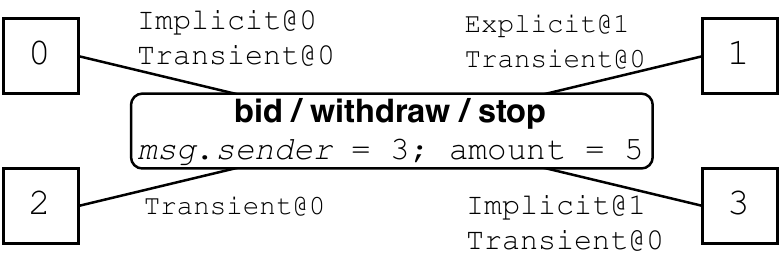}
        \caption{A PT for $4$ users and a fixed action.}
        \label{Fig:Topology:PT}
    \end{subfigure}
    \hfill
    \begin{subfigure}{.49\linewidth}
        \centering
        \includegraphics[scale=0.75]{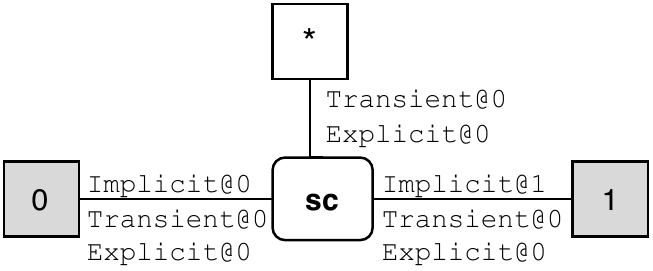}
        \caption{The PTG from \code{PTGBuilder}.}
        \label{Fig:Topology:PTG}
    \end{subfigure}
    \RawCaption{\caption{A PT of \code{Auction} contrasted with a PTG for
                \code{Auction}.} \label{Fig:Topology}}
\end{figure}

A user $u$ participates in communication during a transaction $f$ whenever the
state of $u$ affects execution of $f$ or $f$ affects a state of $u$. We call
this \emph{influence}. For example, in \cref{Fig:Auction}, the sender influences
\code{withdraw} on line~\ref{line:influence1}. Similarly, \code{withdraw}
influences the sender on line~\ref{line:influence2}. In all cases, the influence
is \emph{witnessed} by the state of the contract and the configuration of users
that exhibit the influence.

Let $\cC$ be a contract, $N \in \nat$ be the network size, $( S, P, f, s_0 ) =
\lts(\cC, N)$, and $p \in P$. A user with address $a \in \nat$ \emph{influences}
transaction $f_p$ if there exists an $s, r, r' \in \control( \cC, [N] )$,
$\vu, \vu', \vcv, \vcv' \in \user( \cC, [N] )^N$, and $i \in [N]$ such that:
\begin{enumerate}
\item $\id( \vu_i ) = a$;
\item $\forall j \in \sqrbrack{N} \cdot \left( \vu_j = \vcv_j \right) \iff
      \left( i \ne j \right)$;
\item $( r, \vu' ) = f_p( s, \vu )$ and $( r', \vcv' ) = f_p( s, \vcv )$;
\item $\left( r = r' \right) \implies \left( \exists j \in \sqrbrack{N} \setminus
      \{ i \} \cdot  \vu'_j \neq \vcv'_j \right)$.
\end{enumerate}
That is, there exists two network configurations that differ only in the state of
the user $\vu_i$, and result in different network configurations after applying $f_p$. In practice,
$f_p$ must compare the address of $\vu_i$ to some other address, or must use the state of
$\vu_i$ to determine the outcome of the transaction. The tuple $( s, \vu, \vcv )$ is a
\emph{witness} to the influence of $a$ over transaction $f_p$.
A user with address $a \in \nat$ is \emph{influenced} by transaction $f_p$ if
there exists an $s, s' \in \control( \cC, [N] )$, $\vu, \vu' \in \user( \cC, [N] )^N$, and
$i \in \sqrbrack{N}$ such that:
\begin{enumerate}
\item $\id( \vu_i ) = a$;
\item $( s', \vu' ) = f_p( s, \vu )$;
\item $\vu'_i \ne \vu_i$.
\end{enumerate}
That is, $f_p$ must write into the state of $\vu_i$, and the changes must persist after
the transaction terminates. The tuple $( s, \vu )$ is a \emph{witness} to the influence of
transaction $f_p$ over user $a$.

\begin{definition}[Participation]
\label{Def:Participation}
    A user with address $a \in \nat$ \emph{participates} in a transaction $f_p$
    if either $a$ influences $f_p$, witnessed by some $( s, \vu, \vcv )$, or
    $f_p$ influences $a$, witnessed by some $( s, \vu )$. In either case, $s$
    is a \emph{witness state}. 
\end{definition}

Smart contracts facilitate communication between many users across many
transactions. We need to know every possible participant, and the cause of their
participation---we call this the \emph{participation topology (PT)}. A PT
associates each communication (sending or receiving) with one or more
participation classes, called \emph{explicit}, \emph{transient}, and
\emph{implicit}. The participation is \emph{explicit} if the participant is a
client of the transaction; \emph{transient} if the participant has a role
during the transaction; \emph{implicit} if there is a state such that the
participant is neither a client nor holds any roles. In the case of
MiniSol, all implicit participation is due to literal address values, as
users designated by literal addresses must participate regardless of clients and
roles. An example of implicit participation is when a client is compared to the
address of the zero-account (i.e.,~\code{address(0)}) in \cref{Fig:Auction}.

\begin{definition}[Participation Topology]
    \label{Def:PT}
    A \emph{Participation Topology} of a transaction $f_p$ is a tuple $\pt(\cC,
    N, p) \gets ( \explicit, \transient, \implicit )$, where:
    \begin{enumerate}
    \item $\explicit \subseteq \nat \times [N]$ where $( i, a ) \in \explicit$
          iff $a$ participates during $f_p$, with $\client( p, i ) = a$;
    \item $\transient \subseteq \nat \times [N]$ where  $( i, a ) \in
          \transient$ iff $a$ participates during $f_p$, as witnessed by a state
          $s \in \control( \cC, [N] )$, where $\role( s, i ) = a$;
    \item $\implicit \subseteq [N]$ where $a \in \implicit$ iff $a$ participates
          during $f_p$, as witnessed by a state $s \in \control( \cC, [N] )$,
          where $\forall i \in \nat$, $\role( s, i ) \ne a$ and $\client( p, i )
          \ne a$.
    \end{enumerate}
\end{definition}

For example, \cref{Fig:Topology:PT} shows a PT for any function of
\cref{Fig:Auction} with $4$ users. From \cref{Sec:Intro}, it is clear that each
function can have an affect. The zero-account and smart contract account are
both implicit participants, since changing either account's address to 3 would
block the affect of the transaction. The manager is a transient participant and
the sender is an explicit participant, since the (dis)equality of their
addresses is asserted at lines~\ref{line:neq1}, \ref{line:neq2},
and~\ref{line:eq1}.

\cref{Def:PT} is semantic and dependent on actions. A syntactic summary of all
PTs for all actions is required to reason about communication. This summary is
analogous to over-approximating control-flow with a ``control-flow
graph''~\cite{Allen1970}. This motivates the \emph{Participation Topology Graph
(PTG)} that is a syntactic over-approximation of all possible PTs, independent
of network size. A PTG has a vertex for each user and each action, such that
edges between vertices represent participation classes. In general, a single
vertex can map to many users or actions.

\newcommand{\ap}{\textit{AP}}
PTG edges are labeled by participation classes. For any contract $\cC$, there
are at most $m$ explicit classes and $n$ transient classes, where $n$ is the
number of roles, and $m$ is the maximum number of clients taken by any
function of $\cC$. On the other hand, the number of implicit classes is
determined by the PTG itself. In general, there is no bound on the number of
implicit participants, and it is up to a PTG to provide an appropriate
abstraction (i.e., $L$ in \cref{def:PTG}). The label set common to all PTGs is
$\ap(\cC) := \{
\textit{explicit}@i \mid i \in \sqrbrack{n} \} \cup \{ \textit{transient}@i \mid
i \in \sqrbrack{m} \}$.

\begin{definition}[Participation Topology Graph]
    \label{def:PTG}
    Let $L$ be a finite set of implicit classes, $V \subsetneq \nat$ be finite,
    $E \subseteq V \times V$, and $\delta \subseteq E \times \left( \ap( \cC )
    \cup L \right)$. A \emph{PT Graph} for a contract $\cC$ is a tuple $( ( V,
    E, \delta ), \rho, \tau )$, where $( V, E, \delta )$ is a graph labeled by
    $\delta$, $\rho \subseteq \inputs( \cC, \nat ) \times V$, and $\tau
    \subseteq \inputs( \cC, \nat ) \times \nat \times V$, such that for all $N
    \in \nat$ and for all $p \in \inputs( \cC, [N] )$, with $\pt( \cC, N, p ) =
    ( \explicit, \transient, \implicit )$:
    \begin{enumerate}
    \item If $( i, a ) \in \explicit$, then there exists a $( p, u ) \in \rho$
          and $( p, a, v ) \in \tau$ such that $( u, v ) \in E$ and
          $\delta\left( ( u, v ), \textit{explicit}@i \right)$;
    \item If $( i, a ) \in \transient$, then there exists a $( p, u ) \in \rho$
          and $( p, a, v ) \in \tau$ such that $( u, v ) \in E$ and
          $\delta\left( ( u, v ), \textit{transient}@i \right)$;
    \item If $a \in \implicit$, then there exists a $( p, u ) \in \rho$, $( p,
          a, v ) \in \tau$, and $l \in L$ such that $( u, v ) \in E$ and
          $\delta\left( (u, v), l \right)$.
    \end{enumerate}
\end{definition}

In \cref{def:PTG}, $\tau$ and $\rho$ map actions and users to vertices,
respectively. An edge between an action and a user indicates the potential for
participation. The labels describe the potential participation classes. As an
example, \cref{Fig:Topology:PTG} is a PTG for \cref{Fig:Auction}, where all
actions map to $\mathit{sc}$, the zero-account maps to vertex $0$, the smart
contract account maps to vertex $1$, and all other users map to $\star$. The
two implicit classes have the label $\textit{implicit}@0$ and
$\textit{implicit}@1$, respectively.

\begin{theorem}
    Let $\cC$ be a contract with a PTG $( G, \rho, \tau )$, $G = ( V, E, \delta
    )$, and $\delta \subseteq E \times \left( \ap(\cC) \cup L \right)$. Then,
    for all $N \in \nat$ and all $p \in \inputs( \cC,[N] )$, $\pt( \cC, N, p ) =
    ( \explicit, \transient, \implicit )$ is over-approximated by $( G, \rho,
    \tau )$ as follows:
    \begin{enumerate}
    \item If $\explicit( i, a )$, then $\exists ( u, v ) \in E \cdot \rho( p, u
          ) \land \tau( p, a, v ) \land \delta \left( (u, v),
          \textit{explicit}@i \right)$;
    \item If $\transient( i, a )$, then $\exists ( u, v ) \in E \cdot \rho( p, u
          ) \land \tau( p, a, v ) \land \delta \left( (u, v),
          \textit{transient}@i \right)$;
    \item If $\implicit( a )$, then $\exists ( u, v ) \in E \cdot \exists l \in
          L \cdot \rho( p, u ) \land \tau( p, a, v ) \land \delta \left( (u, v),
          l \right)$.
    \end{enumerate}
\end{theorem}

For any PT, there are many over-approximating PTGs. The weakest PTG joins every
user to every action using all possible labels and a single implicit class.
\cref{Fig:Topology:PTG}, shows a simple, yet stronger, PTG for
\cref{Fig:Auction}. First, note that there are two implicit participants,
identified by addresses 0 and 1, with labels $\textit{implicit}@0$ and
$\textit{implicit}@1$, respectively. Next, observe that any arbitrary user can
become the manager. Finally, the distinctions between actions are ignored. Thus,
there are three user vertices, two which are mapped to the zero-account and
smart contract account, and another mapped to all other users. Such a PTG
is constructed automatically using an algorithm named \code{PTGBuilder}.

\code{PTGBuilder} takes a contract $\cC$ and returns a PTG. The implicit classes
are $L \gets \{ \textit{implicit}@a \mid a \in \nat \}$, where
$\textit{implicit}@a$ signifies implicit communication with address $a$. PTG
construction is reduced to taint analysis~\cite{Kildall1973}. Input address variables, state
address variables, and literal addresses are tainted sources. Sinks are memory
writes, comparison expressions, and mapping accesses. \code{PTGBuilder} computes
$( \mathit{Args}, \mathit{Roles}, \mathit{Lits} )$, where
\begin{inparaenum}[(1)]
\item $\mathit{Args}$ is the set of indices of input variables that propagate to
      a sink;
\item $\mathit{Roles}$ is the set of indices of state variables that propagate
      to a sink;
\item $\mathit{Lits}$ is the set of literal addresses that propagate to a sink.
\end{inparaenum}
Finally, a PTG is constructed as $( G, \rho, \tau )$, where $G = ( V, E, \delta
)$, $\rho \subseteq \inputs( \cC, \nat ) \times V$, $\tau \subseteq \inputs(
\cC, \nat ) \times \nat \times V$, $\textit{sc}$, and $\star$ are unique
vertices:
{\small \begin{enumerate}
\item $V \gets \mathit{Lits} \cup \{ \textit{sc}, \star \} $ and $E \gets \{ (
      \textit{sc}, v ) \mid v \in V \backslash \{ \textit{sc} \} \}$;
\item $\delta \gets \{(e, \textit{explicit}@i) \mid e \in E, i \in \mathit{Args}\}
       \cup
       \{(e, \textit{transient}@i) \mid e \in E, i \in \mathit{Roles} \}
       \cup
       \{((\textit{sc}, a), \textit{transient}@a) \mid a \in \textit{Lits}\}$;
\item $\rho \gets \{ ( p, \textit{sc} ) \mid p \in \inputs( \cC,\nat ) \}$;
\item $\tau \gets \{ ( p, a, \star ) \mid p \in \inputs(\cC,\nat), a \in \nat
      \backslash \mathit{Lits} \} \cup \{ ( p, a, a ) \mid p \in \inputs( \cC,
      \nat), a \in \mathit{Lits} \}$.
\end{enumerate}}

\code{PTGBuilder} 
formalizes the intuition of \cref{Fig:Topology:PTG}. Rule~1 ensures that every
literal address has a vertex, and that all user vertices connect to
$\textit{sc}$.
Rule~2 over-approximates explicit, transient, and implicit labels.
The first set states that if an input address is never used, then the client is not an explicit participant.
This statement is self-evident, and over-approximates explicit participation.
The second and third set make similar claims for roles and literal addresses
Rules~3~and~4 define $\rho$ and
$\tau$ as expected. Note that in MicroSol, implicit participation stems from
literal addresses, since addresses do not support arithmetic operations, and since
numeric expressions cannot be cast to addresses.

By re-framing smart contracts with rendezvous synchronization, each transaction
is re-imagined as a communication between several users. Their communication
patterns are captured by the corresponding PT. A PTG over-approximates PTs of
all transactions, and is automatically constructed using \code{PTGBuilder}. This
is crucial for PCMC as it provides an upper bound on the number of equivalence
classes, and the users in each equivalence class (see
\appendixcite{Appendix:Uniformity}).

\section{Local Reasoning in Smart Contracts}
\label{sec:locality}
In this section, we present a proof rule for the parameterized safety of
MicroSol programs. Our proof rule extends the existing theory of PCMC. The
section is structured as follows. \cref{Sec:Sub:Properties} introduces syntactic
restrictions, for properties and interference invariants, that expose address
dependencies. \cref{Sec:Sub:Localization}, defines local bundle reductions, that
reduce parameterized smart contract models to finite-state models. We show
that for the correct choice of local bundle reduction, the safety of the
finite-state model implies the safety of the parameterized model.

\subsection{Guarded Properties and Split Invariants}
\label{Sec:Sub:Properties}
Universal properties and interference invariants might depend on user addresses.
However, PCMC requires explicit address dependencies. This is because address
dependencies allow predicates to distinguish subsets of users. To resolve this,
we introduce two syntactic forms that make address dependencies explicit:
guarded universal safety properties and split interference invariants. We build
both forms from so called \emph{address-oblivious} predicates that do not depend
on user addresses.

For any smart contract $\cC$ and any address space $\cA$, a pair of user
configurations, $\vu, \vcv \in \user( \cC, \cA )^k$, are \emph{$k$-address
similar} if $\forall i \in [k] \cdot \rec( \vu_i ) = \rec( \vcv_i )$. A
predicate $\xi \subseteq \control( \cC, \cA ) \times \user( \cC, \cA )^k$ is
address-oblivious if, for every choice of $s \in \control( \cC, \cA )$, and
every pair of $k$-address similar configurations, $\vu$ and $\vcv$, $\xi( s, \vu
) \iff \xi( s, \vcv )$. \prop{1} and \prop{2} in \cref{Sec:Intro} are
address-oblivious.

A \emph{guarded $k$-universal safety property} is built from a single $k$-user
address-oblivious predicate. The predicate is guarded by constraints over its
$k$ user addresses. Each constraint compares a single user's address to either
a literal address or a role. This notion is formalized by
\cref{Def:GuardedSafety}, and illustrated in \cref{Ex:GuardedSafety}.

\begin{definition}[Guarded Universal Safety]
    \label{Def:GuardedSafety}
    For $k \in \nat$, a \emph{guarded $k$-universal
    safety property} is a $k$-universal safety property $\varphi$, given by a
    tuple $( L, R, \xi )$, where $L \subsetneq \nat \times [k]$ is finite, $R
    \subsetneq \nat \times [k]$ is finite, and $\xi$ is an address-oblivious
    $k$-user predicate, such that:
    \begin{align*}
        \varphi \left( s, \vu \right)
        &:= \left(
            \left( \Land_{( a, i ) \in L} a = \id( \vu_{i} ) \right)
            \land
            \left( \Land_{( i, j ) \in R} \role( s, i ) = \id( \vu_{j} ) \right)
        \right) \implies \xi( s, \vu )
    \end{align*}
    Note that $\cA_L \gets \{ a \mid ( a, i ) \in L \}$ and $\cA_R \gets \{ i
    \mid ( i, j ) \in R \}$ and  define the \emph{literal} and \emph{role
    guards} for $\varphi$.
\end{definition}

\begin{example}
    \label{Ex:GuardedSafety}
    Consider the claim that in \code{Auction} of \cref{Fig:Auction}, the
    zero-account cannot have an active bid. This claim is stated as \prop{3}:
    \emph{For each user process $\vu$, if $\id( \vu_0 ) = 0$, then $\rec( \vu_0
    )_0 = 0$}. That is, \prop{3}  is a guarded $1$-universal safety property 
    $\varphi_1(s, \vu) \gets \left( 0 =  \id( \vu_0 ) \right) \implies \left(
    \rec( \vu_0 )_0 = 0 \right)$. Following \cref{Def:GuardedSafety},
    $\varphi_1$ is determined by $( L_1, \varnothing, \xi_1 )$, where $L_1 = \{
    ( 0, 0 ) \}$ and $\xi_1( s, \vu ) \gets \rec( \vu_0 )_0 = 0$. The second set
    is $\varnothing$ as there are no role constraints in \prop{3}. If a state
    $(s, \vu)$ satisfies $\varphi_1$, then $\forall \{ i \} \subseteq [N] \cdot
    \varphi_1( s, ( \vu_i ) )$. Note that $\vu$ is a singleton vector, and that
    $\varphi_1$ has $1$ literal guard, given by $\{ 0 \}$. \qed
\end{example}

The syntax of a \emph{split interference invariant} is similar to a guarded
safety property. The invariant is constructed from a list of address-oblivious
predicates, each guarded by a single constraint. The final predicate is guarded
by the negation of all other constraints. Intuitively, each address-oblivious
predicate summarizes the class of users that satisfy its guard. The split
interference invariant is the conjunction of all (guarded predicate) clauses.
We proceed with the formal definition in \cref{Def:CompInvar} and a practical
illustration in \cref{Ex:SplitCompInvar}.

\begin{definition}[Split Interference Invariant]
    \label{Def:CompInvar}
    A \emph{split interference invariant} is an
    interference invariant $\theta$, given by a tuple $( \cA_L, \cA_R, \vec{\zeta},
    \vec{\mu}, \xi )$, where $\cA_L = \{ l_0, \ldots, l_{m-1} \} \subsetneq \nat$ is
    finite, $\cA_R = \{ r_0, \ldots, r_{n-1} \} \subsetneq \nat$ is finite,
    $\vec{\zeta}$ is a list of $m$ address-oblivious $1$-user predicates,
    $\vec{\mu}$ is a list of $n$ address-oblivious $1$-user predicates, and
    $\xi$ is an address-oblivious $1$-user predicate, such that:
    \begin{align*}
        \psi_{\mathrm{Lits}}(s, \vu)
            &:= \left( \Land_{i = 0}^{m-1} \id(\vu_0) = l_i \right)
                \implies \vec{\zeta}_i(s, \vu) \\
        \psi_{\mathrm{Roles}}(s, \vu)
            &:= \left( \Land_{i = 0}^{n-1} \id(\vu_0) = \role(s, r_i) \right)
                \implies \vec{\mu}_i(s, \vu) \\
        \psi_{\mathrm{Else}}(s, \vu)
            &:= \left(\left( \Land_{i = 0}^{m-1} \id(\vu_0) \ne l_i \right)
                \land \left( \Land_{i = 0}^{n-1} \id(\vu_0) \ne \role(s, r_i) \right)\right)
                \implies \xi(s, \vu) \\
        \theta(s, \vu)
            &:= \psi_{\mathrm{Roles}}(s, \vu)
                \land \psi_{\mathrm{Lits}}(s, \vu) 
                \land \psi_{\mathrm{Else}}(s, \vu)
    \end{align*}
    Note that $\cA_L$ and $\cA_R$ define \emph{literal} and \emph{role guards}
    of $\theta$, and that $|\vu| = 1$.
\end{definition}

\begin{example}
    \label{Ex:SplitCompInvar}
    To establish $\varphi_1$ from \cref{Ex:GuardedSafety}, we require an
    adequate interference invariant such as \prop{4}: \emph{The zero-account
    never has an active bid, while all other users can have active bids}. That
    is, \prop{4} is a split interference invariant:
    {\small \begin{equation*}
        \theta_1(s, \vu) \gets\;
            \left( \id(\vu_0) = 0 \;\;\implies\; (\rec(\vu_0))_0 = 0 \right)
            \;\land\;
            \left( \id(\vu_0) \ne 0 \;\;\implies\; (\rec(\vu_0))_0 \ge 0 \right)
    \end{equation*}}
    Following \cref{Def:CompInvar}, $\theta_1$ is determined by $\mathit{Inv} = (
    \cA_L, \varnothing, ( \xi_1 ), \varnothing, \xi_2 )$, where $\cA_L = \{ 0 \}$,
    $\xi_1$ is defined in \cref{Ex:GuardedSafety}, and $\xi_2( s, \vu ) \gets
    \rec( \vu_0 )_0 \ge 0$. The two instances of $\varnothing$ in $\mathit{Inv}$
    correspond to the lack of role constraints in $\theta_1$. If $\mathit{Inv}$
    is related back to \cref{Def:CompInvar}, then $\psi_{\mathrm{Roles}}( s, \vu
    ) \gets \top$, $\psi_{\mathrm{Lits}}( s, \vu ) \gets \left( \id(\vu_0) = 0
    \right) \implies \left( \rec( \vu_0 )_0 = 0 \right)$, and
    $\psi_{\mathrm{Else}}( s, \vu ) \gets \left( \id( \vu_0 ) \ne 0 \right)
    \implies \left( \rec( \vu_0)_0 \ge 0 \right)$. \qed
\end{example}

\subsection{Localizing a Smart Contract Bundle}
\label{Sec:Sub:Localization}
A local bundle is a finite-state abstraction of a smart contract bundle. This
abstraction reduces smart contract PCMC to software model checking. At a high
level, each local bundle is a non-deterministic LTS and is constructed from
three components: a smart contract, a candidate interference invariant, and a
neighbourhood. The term \emph{candidate interference invariant} describes any
predicate with the syntax of an interference invariant, regardless of its
semantic interpretation. Sets of addresses are used to identify representatives
in a neighbourhood.

Let $\cA$ be an $N$-user neighbourhood and $\theta_U$ be a candidate
interference invariant. The local bundle corresponding to $\cA$ and $\theta_U$
is defined using a special relation called an \emph{$N$-user interference
relation}. The $N$-user interference relation (for $\theta_U$) sends an $N$-user
smart contract state to the set of all $N$-user smart contract states that are
reachable under the interference of $\theta_U$. A state is reachable under
the interference of $\theta_U$ if the control state is unchanged, each address
is unchanged, and all user data satisfies $\theta_U$. For example, lines \ref{line:harness-instr-start}--\ref{line:harness-instr-end}
in \cref{Fig:Harness} apply a $4$-user interference relation to the states of
\code{Auction}. Note that if the interference relation for $\theta_U$ fails to
relate $(s, \vu)$ to itself, then $(s, \vu)$ violates $\theta_U$.
\begin{definition}[Interference Relation]
    \label{Def:InterfRel}
    Let $N \in \nat$, $\cC$ be a contract, $S = \control( \cC, \nat ) \times
    \user( \cC, \nat )^N$, and $\theta_U$ be a split candidate interference
    invariant. The $N$-user interference relation for $\theta_U$ is the relation
    $g: S \rightarrow 2^S$ such that $g( c, \vu ) \gets \{ ( c, \vcv ) \in S
    \mid \forall i \in [N] \cdot \id( \vu_i ) = \id( \vcv_i ) \land \theta_U( s,
    \vcv_i ) \}$.
\end{definition}

Each state of the \emph{local bundle} for $\cA$ and $\theta_U$ is a tuple $( s,
\vu )$, where $s$ is a control state and $\vu$ is an $N$-user configuration. The
$N$ users in the local bundle correspond to the $N$ representatives in $\cA$,
and therefore, the address space of the local bundle can be non-consecutive. The
transition relation of the local bundle is defined in terms of the (global)
transaction function $f$. First, the transition relation applies $f$. If the
application of $f$ is closed under $\theta_U$, then the interference relation is
applied. Intuitively, $\theta_U$ defines a safe envelop under which the
interference relation is compositional.
\begin{definition}[Local Bundle]
    \label{Def:LocalBundle}
    Let $\cC$ be a contract, $\cA = \{ a_0, \ldots, a_{N-1} \}\subseteq \nat$ be
    an $N$-user neighbourhood, $\theta_U$ be a candidate split interference
    invariant, and $g$ be the $N$-user interference relation for $\theta_U$. A
    \emph{local bundle} is an LTS $\local(\cC, \cA, \theta_U) \gets ( S, P,
    \hat{f}, s_0 )$, such that $S \gets \control( \cC, \cA ) \times \user( \cC,
    \cA )^N$, $P \gets \inputs( \cC, \cA )$, $s_0 \gets ( c_0, \vu )$, $c_0 \gets
    ( \vec{0}, \vec{0} )$, $\forall i \in [N] \cdot \id( \vu_i ) = a_i \land
    \rec( \vu_i ) = \vec{0}$, and $\hat{f}$ is defined with respect to $\cM: \cA
    \rightarrow [N]$, $\cM(a_i) = i$, such that:
    {\begin{equation*}
        \hat{f}((s, \vu), p) \gets \begin{cases}
            g(s', \vu')
                & \text{if } (s', \vu') = \bbrackfn{\cC}_{\cM}((s, \vu), p)
                  \land (s', \vu') \in g(s', \vu') \\
            \bbrackfn{\cC}_{\cM}((s, \vu), p) & \text{otherwise}
        \end{cases}
    \end{equation*}}
\end{definition}

\begin{example}
    \label{Ex:LocalBundle}
    We briefly illustrate the transition relation of \cref{Def:LocalBundle}
    using \code{Auction} of \cref{Fig:Auction}. Let $\cA_1 = \{ 0, 1, 2, 3 \}$
    be a neighbourhood, $\theta_1$ be as in \cref{Ex:SplitCompInvar}, $g$ be the
    $4$-user interference relation for $\theta_1$, and $( S, P, \hat{f}, s_0 ) =
    \local( \cC, \cA_1, \theta_1 )$. Consider applying $\hat{f}$ to $( s, \vu )
    \in S$ with action $p \in P$, such that $s = \{ \mathsol{manager} \mapsto 2;
    \mathsol{leadingBid} \mapsto 0 \}$, $\forall i \in [4] \cdot \rec( \vu_i ) =
    0$, and $p$ is a bid of $10$ from a sender at address $3$.
    
    By definition, if $( s', \vcv ) = f( s, \vu, p )$, then the leading bid is
    now $10$, and the bid of the sender is also $10$, since the sender of $p$ was
    not the manager and the leading bid was less than $10$. Clearly $( s', \vcv )
    \in g( s', \vcv )$, and therefore, $g( s', \vcv ) = \hat{f}\left( ( s, \vu ), p
    \right)$. A successor state is then selected, as depicted in
    \cref{Fig:Local:Sub:Bundle}. This is done by first assigning an arbitrary
    bid to each representative, and then requiring that each bid satisfies
    $\theta_1$ relative to $s'$. In \cref{Fig:Local:Sub:Bundle}, a network is
    selected in which $\forall i \in [4] \cdot \id(\vcv_i) = i$. As depicted in
    \cref{Fig:Local:Sub:Bundle}, $\theta_1$ stipulates that the zero-account
    must satisfy $\xi_1$ and that all other users must satisfy $\xi_2$.

    In \cref{Fig:Local:Sub:AddrSet}, a satisfying bid is assigned to each user. The
    choice for $d_0$ was fixed since $\xi_1( s, \vcv_0 )$ entails $d_0 = 0$. For 
    $d_1$ to $d_3$, any non-negative value could have been selected. After the
    transaction is executed, $\rec( \vu'_0 )_0 = 0$, $\rec( \vu'_1 )_0 = 1$,
    $\rec( \vu'_2 )_0 = 2$, $\rec( \vu'_3 )_0 = 3$, and $s' = \{
    \mathsol{manager} \mapsto 2; \mathsol{leadingBid} \mapsto 10 \}$. Then $(
    s', \vu' ) \in \hat{f}( s, \vu )$, as desired. Note that $( s', \vu' )$ is
    not reachable in $\lts( \cC, 4 )$. \qed
\end{example}

\begin{figure}[t]
    \centering
    \begin{subfigure}{.49\linewidth}
        \centering
        \includegraphics[scale=0.65]{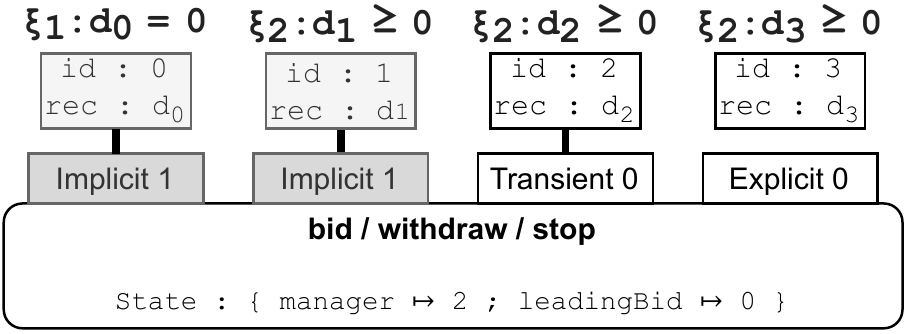}
        \caption{A local $4$-user configuration.}
        \label{Fig:Local:Sub:Bundle}
    \end{subfigure}
    \hfill
    \begin{subfigure}{.49\linewidth}
        \centering
        \includegraphics[scale=0.65]{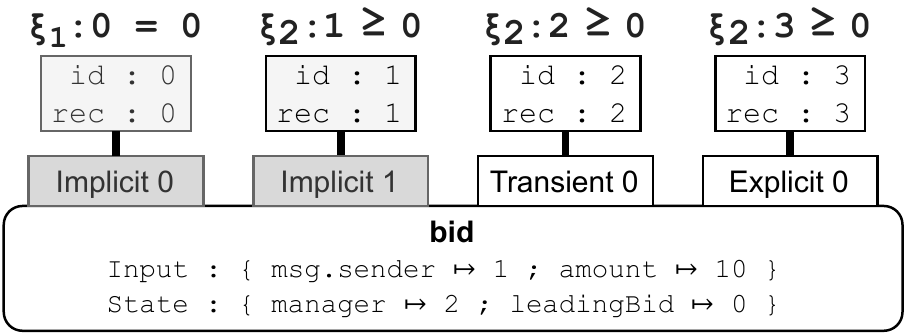}
        \caption{The saturating property of $\cA_1$.}
        \label{Fig:Local:Sub:AddrSet}
    \end{subfigure}

    \hfill
    \RawCaption{\caption{The local bundle for \code{Auction} in
                         \cref{Fig:Auction}, as defined by $\cA_1$ and
                         $\theta_1$ in \cref{Ex:LocalBundle}.}
                \label{Fig:Local}}
\end{figure}

\cref{Ex:LocalBundle} motivates an important result for local bundles. Observe
that $( s', \vu' ) \models \theta_1$. This is not by chance. First, by the
compositionality of $\theta_1$, all user configurations reached by $\local( \cC,
\cA_1, \theta_1 )$ must satisfy $\theta_U$. Second, and far less obviously, by
choice of $\cA_1$, if all reachable user configurations satisfy $\theta_1$,
then $\theta_1$ must be compositional. The proof of this result relies on a
saturating property of $\cA_1$.

A neighbourhood $\cA$ is \emph{saturating} if it contains representatives from
each participation class of a PTG, and for all role guards ($\cA_R \subsetneq
\nat$) and literal guards ($\cA_L \subseteq \nat$) of interest. Intuitively,
each participation class over-approximates an equivalence class of $\cC$. The
number of representatives is determined by the equivalence class. In the case of
\code{PTGBuilder}, a saturating neighbourhood contains one address for each
participation class. For an implicit class, such as $\textit{implicit}@x$, $x$
is literal and must appear in the neighbourhood. All other addresses are
selected arbitrarily. The saturating property of $\cA_1$ is depicted in
\cref{Fig:Local:Sub:AddrSet} by the correspondence between users and
participation classes ($\cA_R = \varnothing$, $\cA_L = \{ 0 \}$).

\newcommand{\Aexp}{\cA_{\mathrm{Exp}}}
\newcommand{\Atrans}{\cA_{\mathrm{Trans}}}
\newcommand{\Aimpl}{\cA_{\mathrm{Impl}}}
\begin{definition}[Saturating Neighbourhood]
    Let $\cA_R, \cA_L \subseteq \nat$, $\cC$ be a contract, $( G, \rho, \tau )$
    be the \code{PTGBuilder} PTG of $\cC$, and $G = ( V, E, \delta )$ such that
    $\cA_R$ and $\cA_L$ are finite. A \emph{saturating neighbourhood} for $(
    \cA_R, \cA_L, ( G, \rho, \tau ) )$ is a set $\Aexp \cup \Atrans \cup \Aimpl$
    s.t. $\Aexp, \Atrans, \Aimpl \subseteq \nat$ are pairwise disjoint and:
    {\begin{enumerate}
    \item $|\Aexp| = |\{ i \in \nat \mid \exists e \in E \cdot \delta \left( e,
          \textit{explicit}@i \right) \}|$,
    \item $|\Atrans| = |\{ i \in \nat \mid \exists e \in E \cdot \delta \left(
          e, \textit{transient}@i \right) \} \cup \cA_R|$,
    \item $\Aimpl = \{ x \in \nat \mid \exists e \in E \cdot \delta \left( e,
          \textit{implicit}@x \right) \} \cup \cA_L$.
    \end{enumerate}}
\end{definition}

A saturating neighbourhood can be used to reduce compositionality and $k$-safety
proofs to the safety of local bundles. We start with compositionality. Consider
a local bundle with a neighbourhood $\cA^{+}$, where $\cA^{+}$ contains a
saturating neighbourhood, the guards of $\theta_U$, and some other address $a$.
The neighbourhood $\cA^{+}$ contains a representative for: each participation
class; each role and literal user distinguished by $\theta_U$; an arbitrary user
under interference (i.e., $a$).
We first claim that if $\theta_U$ is compositional, then a local bundle
constructed from $\theta_U$ must be safe with respect to $\theta_U$ (as in
\cref{Ex:LocalBundle}). The first claim follows by induction. By
\textbf{Initialization} (\cref{Sec:Background}), the initial users satisfy
$\theta_U$. For the inductive step, assume that all users satisfy $\theta_U$ and
apply $\hat{f}_p$. The users that participate in $\hat{f}_p$ maintain $\theta_U$
by \textbf{Consecution} (\cref{Sec:Background}). The users that do not
participate also maintain $\theta_U$ by \textbf{Non-Interference}
(\cref{Sec:Background}). By induction, the first claim is true.
We also claim that for a sufficiently large neighbourhood---say $\cA^{+}$---the
converse is also true. Intuitively, $\cA^{+}$ is large enough to represent each
equivalence class imposed by both the smart contract and $\theta_U$, along with
an arbitrary user under interference. Our key insight is that the reachable
control states of the local bundle form an inductive invariant $\theta_C$. If
the local bundle is safe, then the interference relation is applied after each
transition, and, therefore, the local bundle considers every pair of control and
user states $(c, u)$ such that $c \in \theta_C$ and $(c, u) \in \theta_U$.
Therefore, the safety of the local bundle implies \textbf{Initialization},
\textbf{Consecution}, and \textbf{Non-Interference}.
This discussion justifies \cref{Thm:CompCheck}.

\begin{theorem}
    \label{Thm:CompCheck}
    Let $\cC$ be a contract, $G$ be a PTG for $\cC$, $\theta_U$
    be a candidate split interference invariant with role guards $\cA_R$ and
    literal guards $\cA_L$, $\cA$ be a saturating neighbourhood for $( \cA_R,
    \cA_L, G )$, $a \in \nat \backslash \cA$, and $\cA^{+} = \{ a \} \cup \cA$.
    Then, $\local(\cC, \cA^{+}, \theta_U) \models \theta_U$ if and only if
    $\theta_U$ is an interference invariant for $\cC$.
\end{theorem}

Next, we present our main result: a sound proof rule for $k$-universal safety.
As in \cref{Thm:CompCheck}, \cref{Thm:SafetyCheck} uses a saturating neighbourhood
$\cA^{+}$. This proof rule proves inductiveness, rather than compositionality,
so $\cA^{+}$ does not require an arbitrary user under interference. However, a
$k$-universal property can distinguish between $k$ users at once. Thus,
$\cA^{+}$ must have at least $k$ arbitrary representatives.

\begin{theorem}
    \label{Thm:SafetyCheck}
    Let $\varphi$ be a $k$-universal safety property with role guards $\cA_R$
    and literal guards $\cA_L$, $\cC$ be a contract,
    $\theta_U$ be an interference invariant for $\cC$, $G$ be a PTG for $\cC$,
    $\cA = \Aexp \cup \Atrans \cup \Aimpl$ be a saturating neighbourhood for $(
    \cA_R, \cA_L, G )$. Define $\cA^{+} \subseteq \nat$ such that $\cA \subseteq
    \cA^{+}$ and $|\cA^{+}| = |\cA| + \max( 0, k - |\Aexp| )$. If $\local( \cC,
    \cA^{+}, \theta_U ) \models \varphi$, then $\forall N \in \nat \cdot \lts(
    \cC, N ) \models \varphi$.
\end{theorem}

\cref{Thm:SafetyCheck} completes \cref{Ex:SplitCompInvar}. Recall $( \varphi_1,
\theta_1, \cA_1 )$ from \cref{Ex:LocalBundle}. Since $\varphi_1$ is
$1$-universal and $\cA_1$ has one explicit representative, it follows that
$\cA^{+} = \cA_1 \cup \varnothing$. Using an SMC, $\local( \cC, \cA_1^{+},
\theta_1 ) \models \varphi_1$ is certified by an inductive strengthening
$\theta_1^{*}$. Then by \cref{Thm:SafetyCheck}, $\cC$ is also safe for $2^{160}$
users. Both the local and global bundle have states exponential in the number of
users. However, the local bundle has $4$ users (a constant fixed by $\cC$),
whereas the global bundle is defined for any number of users. This achieves an
exponential state reduction with respect to the network size. Even more
remarkably, $\theta_1^{*}$ must be the inductive invariant from
\cref{Sec:Background}, as it summarizes the safe control states that are closed
under the interference of $\theta_1$. Therefore, we have achieved an exponential
speedup in verification and have automated the discovery of an inductive
invariant.


\section{Implementation and Evaluation}
\label{sec:eval}
We implement smart contract PCMC as an open-source tool called \smartace, that
is built upon the Solidity compiler. It works in the following automated steps:
(1) consume a Solidity smart contract and its interference invariants; (2)
validate the contract's conformance to MiniSol; (3) perform source-code analysis
and transformation (i.e., inheritance inlining, devirtualization,
\code{PTGBuilder}); (4) generate a local bundle in LLVM IR; (5) verify the
bundle using \seahorn~\cite{GurfinkelKahsai2015}.
In this section, we report on the effectiveness of \smartace in verifying
real-world smart contracts. A full description of the \smartace architecture and
of each case study is beyond the scope of this paper. Both \smartace and the case
studies are available\footnote{\anonymize{\url{\contractaceurl}}}.
Our evaluation answers the following research questions:
\begin{description}
\item[RQ1: Compliance.] Can MiniSol represent real-world smart contracts?
\item[RQ2: Effectiveness.] Is \smartace effective for MiniSol smart contracts?
\item[RQ3: Performance.] Is \smartace competitive with other techniques?
\end{description}

\paragraph{Benchmarks and Setup.}
To answer the above research questions, we used a benchmark of 89 properties
across 15 smart contracts (see~\cref{Table:Benchmarks}). Contracts
\code{Alchemist} to \code{Mana} are from \verx~\cite{PermenevDimitrov2020}.
Contracts \code{Fund} and \code{Auction} were added to offset the lack of
parameterized properties in existing benchmarks. The \code{QSPStaking} contract
comprises the Quantstamp Assurance Protocol\footnote{\scriptsize \url{\qspurl}} for which we
checked real-world properties provided by Quantstamp. Some properties require
additional instrumentation techniques (i.e.,
temporal~\cite{PermenevDimitrov2020} and aggregate~\cite{HajduJovanovic2019}
properties). Aggregate properties allow \smartace to reason about the sum of all
records within a mapping. In \cref{Table:Benchmarks}, \emph{Inv. Size} is the
clause size of an interference invariant manually provided to \smartace and
\emph{Users} is the maximum number of users requested by \code{PTGBuilder}. All
experiments were run on an Intel\trademark~Core i7\trademark~CPU~@~2.8GHz 4-core
machine with 16GB of RAM~on~Ubuntu~18.04.

\begin{table*}[t]
  \small
  \centering
  \begin{tabular}{@{}lrrrrrrcr@{}}
  \toprule \multicolumn{3}{c}{Contracts} &
  \phantom{abc} & \multicolumn{3}{c}{\smartace} &
  \phantom{abc} & \multicolumn{1}{c}{\verx} \\

  \cmidrule{1-3} \cmidrule{5-7} \cmidrule{9-9}
  Name & Prop. & LOC && Time & Inv. Size & Users && Time \\

  \midrule
  Alchemist & 3 & 401 && 7 & 0 & 7 && 29 \\
  ERC20 & 9 & 599 && 12 & 1 & 5 && 158 \\
  Melon & 16 & 462 && 30 & 0 & 7 && 408 \\
  MRV & 5 & 868 && 2 & 0 & 7 && 887 \\
  Overview & 4 & 66 && 4 & 0 & 8 && 211 \\
  PolicyPal & 4 & 815 && 26 & 0 & 8 && 20,773 \\
  Zebi & 5 & 1,209 && 8 & 0 & 7 && 77 \\
  Zilliqa & 5 & 377 && 8 & 0 & 7 && 94 \\
  \midrule
  Brickblock & 6 & 549 && 13 & 0 & 10 && 191 \\
  Crowdsale & 9 & 1,198 && 223 & 0 & 8 && 261 \\
  ICO & 8 & 650 && 371 & 0 & 16 && 6,817 \\
  VUToken & 5 & 1,120 && 19 & 0 & 10 && 715 \\
  Mana & 4 & 885 && --- & --- & --- && 41,409 \\
  \midrule
  Fund & 2 & 38 && 1 & 0 & 6 && --- \\
  Auction & 1 & 42 && 1 & 1 & 5 && --- \\
  QSPStaking & 4 & 1,550 && 3 & 7 & 8 && --- \\
  \bottomrule
\end{tabular}
  \caption{Experimental results for \smartace. All reported times are in seconds.}
  \label{Table:Benchmarks}
\end{table*}

\paragraph{RQ1: Compliance.}
To assess if the restrictions of MiniSol are reasonable, we find the number of
\emph{compliant} \verx benchmarks. We found that $8$ out of $13$ benchmarks are
compliant after removing dead code. With manual abstraction, $4$ more
benchmarks complied. \code{Brickblock} uses inline assembly to revert transactions
with smart contract senders. We remove the assembly as an over-approximation. To
support \code{Crowdsale}, we manually resolve dynamic calls not supported by
\smartace. In \code{ICO}, calls are made to arbitrary contracts (by address).
However, these calls adhere to \emph{effectively external callback
freedom}~\cite{GershuniAmit2019,PermenevDimitrov2020} and can be omitted. Also,
\code{ICO} uses dynamic allocation, but the allocation is performed once. We
inline the first allocation, and assert that all other allocations are
unreachable. To support \code{VUToken}, we replace a dynamic array of bounded
size with variables corresponding to each element of the array. The function
\code{_calcTokenAmount} iterates over the array, so we specialize each call
(i.e., \code{_calcTokenAmount_\{1,2,3,4\}}) to eliminate recursion. Two other
functions displayed unbounded behaviour (i.e., \code{massTransfer} and
\code{addManyToWhitelist}), but are used to sequence calls to other functions,
and do not impact reachability. We conclude that the restrictions of MiniSol are
reasonable.

\paragraph{RQ2: Effectiveness.}
To assess the effectiveness of \smartace, we determined the number of properties
verified from compliant \verx contracts. We found that all properties could be
verified, but also discovered that most properties were not parameterized. To
validate \smartace with parameterized properties, we conducted a second study
using \code{Auction}, as described on our development
blog\footnote{\anonymize{\url{\seahornblogurl}}}.
To validate \smartace in the context of large-scale contract development, we
performed a third study using \code{QSPStaking}.
In this study, $4$ properties were selected at random, from a specification
provided by Quantstamp, and validated. It required $2$ person days to model the
environment, and $1$ person day to discover an interference invariant. The major
overhead in modeling the environment came from manual abstraction of unbounded
arrays. The discovery of an interference invariant and array abstractions were
semi-automatic, and aided by counterexamples from \seahorn. For example, one
invariant used in our abstraction says that all elements in the array
\code{powersOf100} must be non-zero. This invariant was derived from a
counterexample in which $0$ was read spuriously from \code{powersOf100},
resulting in a division-by-zero error. We conclude that \smartace is suitable for
high-assurance contracts, and with proper automation, can be integrated into
contract development.

\paragraph{RQ3: Performance.}
To evaluate the performance of \smartace, we compared its verification time to
the reported time of \verx, a state-of-the-art, semi-automated verification tool.
Note that in \verx, predicate abstractions must be provided manually, whereas
\smartace automates this step. \verx was evaluated on a faster processor (3.4GHz)
with more RAM (64GB)\footnote{We have requested access to \verx and are awaiting a
response.}. In each case, \smartace significantly outperformed \verx, achieving a
speedup of at least 10x for all but $2$ contracts\footnote{We compare the average
time for \verx to the total evaluation time for \smartace.}.
One advantage of \smartace is that it benefits from state-of-the art software
model checkers, whereas the design of \verx requires implementing a new
verification tool. In addition, we suspect that local bundle abstractions
obtained through smart contract PCMC are easier to reason about than the global
arrays that \verx must quantify over. However, a complete explanation for the
performance improvements of \smartace is challenging without access to the source
code of \verx.
We observe that one bottleneck for \smartace is the number of users (which
extends the state space). A more precise \code{PTGBuilder} would reduce the
number of users. Upon manual inspection of \code{Melon} and \code{Alchemist} (in
a single bundle), we found that user state could be reduced by 28\%. We conclude
that \smartace can scale.

\section{Related Work}
\label{sec:related}

In recent years, the program analysis community has developed many tools for
smart contract analysis. These tool range from dynamic
analysis~\cite{JiangLiu2018,WuestholzChristakis2020} to static
analysis~\cite{LuuChu2016,MossbergManzano2019,KruppRossow2018,TsankovDan2018,GrechKong2018,NikolicKolluri2018,KolluriNikolic2019,WangZhang2019,BrentGrech2020}
and verification~\cite{KalraGoel2018,HajduJovanovic2019,WangLahiri2019,MavridouLaszka2019,PermenevDimitrov2020,SoLee2020}.
The latter are most related to \smartace since their focus is on
functional correctness, as opposed to generic rules (e.g., the absence of
reentrancy~\cite{GrossmanAbraham2018} and integer overflows). Existing
techniques for functional correctness are either deductive, and require that
most invariants be provided manually
(i.e.,~\cite{HajduJovanovic2019,WangLahiri2019}), or are automated but neglect
the parameterized nature of smart contracts
(i.e.,~\cite{MarescottiOtoni2020,MavridouLaszka2019,PermenevDimitrov2020,SoLee2020}). The
tools that do acknowledge parameterization employ static analysis
\cite{KolluriNikolic2019,BrentGrech2020}. In contrast, \smartace uses a
novel local reasoning technique that verifies parameterized safety properties
with less human guidance than deductive techniques.

More generally, parameterized systems form a rich field of research, as outlined
in~\cite{BloemJacobs2015}. The use of SCUNs was first proposed
in~\cite{GermanSistla1992}, and many other models exist for both synchronous and
asynchronous systems
(e.g.,~\cite{EsparzaGanty2013,SiegelAvrunin2004,SiegelGopalakrishnan2011}).
The approach of PCMC is not the only compositional solution for parameterized
verification. For instance, environmental abstraction~\cite{ClarkeTalupur2006}
considers a process and its environment, similar to the inductive and
interference invariants of \smartace. Other
approaches~\cite{PnueliRuah2001,FangPiterman2004} generalize from small
instances through the use of ranking functions. The combination of abstract
domains and SMPs has also proven useful in finding parameterized
invariants~\cite{AbdullaHH16}. The addresses used in our analysis are similar
to the scalarsets of \cite{IpDill1993}. Most compositional techniques require
cutoff analysis---considering network instances up to a given
size~\cite{EmersonNamjoshi1995,KaiserKroening2010,KhalimovJacobs2013}. Local
bundles avoid explicit cutoff analysis by simulating all smaller instances,
and is similar to existing work on bounded parameterized model
checking~\cite{EmersonTrefler2006}. \smartace is the first application of PCMC
in the context of smart contracts.

\section{Conclusions}
\label{sec:conclusions}

In this paper, we present a new verification approach for Solidity smart contracts.
Unlike many of the existing approaches, we automatically reason about smart contracts
relative to all of their clients and accross multiple transaction. Our approach is based
on treating smart contracts as a parameterized system and using Parameterized Compositional Model Checking (PCMC).

Our main theoretical contribution is to show that PCMC offers an exponential reduction for
$k$-universal safety verification of smart contracts. That is, verification of safety properties with 
$k$ arbitrary clients. 

The theoretical results of this paper are implemented in an automated Solidity verification tool \smartace. \smartace is built upon a novel model for smart contracts, in which
users are processes and communication is explicit. In this model, communication
is over-approximated by static analysis, and the results are sufficient to find
all local neighbourhoods, as required by PCMC. The underlying parameterized verification task is reduced to sequential Software Model Checking. In \smartace, we use the \seahorn verification framework for the underlying analysis. However, other Software Model Checkers can potentially be used as well. 

Our approach is almost completely automated -- \smartace automatically infers the necessary predicates, inductive invariants, and transaction summaries. The only requirement from the user is to provide an occasional interference invariant (that is validated by \smartace). However, we believe that this step can be automated as well through reduction to satisfiability of Constrained Horn Clauses. We leave exploring this to future work.


\bibliographystyle{splncs04}
\bibliography{bibliography}


\ifthenelse{\boolean{extended}}{\appendix
\newpage
\section{The Auction LTS}
\label{Appendix:Lts}

In this section we formally define the bundle of $\cC \gets \mathsol{Auction}$,
from \cref{Fig:Auction}. We no longer assume that $\cC$ consists of a single
transaction. This means that each input must also specify the name of the
function to execute. In the case of $\cC$, the function name are defined by
$\textrm{Tx} \gets \{ \texttt{constructor},\; \texttt{bid},\;
\texttt{withdraw},\; \texttt{stop} \}$.

The control states are $\control( \cC, \cA ) \subseteq ( \cA ) \times ( \bbD
\times \bbD ) \times ( \bbD ) $, where $\cA$ is the address space. For each
state $( a_1, d_1, d_2, \mathrm{aux}_1 ) \in \control( \cC, \cA )$, $a_1$,
$d_1$, and $d_2$ correspond to the \code{manager}, \code{leadingBid}, and
\code{stopped}, respectively. An auxiliary variable, $\mathrm{aux}_1$, is added
to indicate whether the constructor has been called. When $\mathrm{aux}_1$ is
$0$, then the constructor has not been called and \emph{must} be called in the
next transaction. When $\mathrm{aux}_1$ is $1$, then the constructor has been
called, and may not be called again. For simplicity, $\data( s, 2 ) =
\mathrm{aux}_1$.

The user states are $\user( \cC, \cA ) \subseteq \cA \times ( \bbD )$. For each
user $( x, y_1 ) \in \control( \cC, \cA )$, $x$ is the user's address and $y_1$
is the user's bid.

The inputs are $\inputs( \cC, \cA ) \subseteq \textrm{Tx} \times ( \cA \times
\cA ) \times ( \bbD )$. For any $p = (t, x_1, x_2, y_1) \in \inputs( \cC, \cA
)$, the interpretation of $p$ depends on $t$. In all cases, $x_1$ is
\code{msg.sender}. If $t = \texttt{constructor}$, then $x_2$ represents
\code{mgr}. If $t = \mathsol{bid}$, then $y_1$ represents \code{amount}. In all
other cases, $x_2$ and $y_1$ are unused.

For the rest of this section, we restrict our discussion to $\lts( \cC, N )$.
This means that $\cA = [N]$, and that $\cM$ is the identify function for $[N]$.
We now present transactional semantics for $f = \bbrackfn{\cC}_{\cM}$. We assume
that $N \ge 2$, so that the zero-account and smart contract account are defined.

\begin{align*}
    f( ( s, \vu ), ( t, x_1, x_2, y_1 ) )
    &=
    \begin{cases}
        ( s, \vu ) & \text{if } x_1 = \id( \vu_0 ) \\
        ( s, \vu ) & \text{if } x_1 = \id( \vu_1 ) \\
        g_1( ( s, \vu ), ( x_1, x_2 ) ) & \text{if } t = \texttt{constructor} \\
        g_2( ( s, \vu ), ( x_1, y_1 ) ) & \text{if } t = \texttt{bid} \\
        g_3( ( s, \vu ), ( x_1 ) ) & \text{if } t = \texttt{withdraw} \\
        g_4( ( s, \vu ), ( x_1 ) ) & \text{if } t = \texttt{stop} \\
    \end{cases} \\
    g_1( ( s, \vu ), ( x_1, x_2 ) )
    &=
    \begin{cases}
        ( s, \vu ) & \text{if } \data( s, 2 ) \ne 0 \\
        ( s', \vu ) & \text{else, s.t. }
                      \role( s', 0 ) = x_2,\;
                      \data( s', 0 ) = \data( s, 0 ), \\&
                      \data( s', 1 ) = \data( s, 1 ), \text{ and }
                      \data( s', 2 ) = 1
    \end{cases} \\
    g_2( ( s, \vu ), ( x_1, y_1 ) )
    &=
    \begin{cases}
        ( s, \vu ) & \text{if } \data( s, 2 ) \ne 1 \\
        ( s, \vu ) & \text{if } \role( s, 0 ) = x_1 \\
        ( s, \vu ) & \text{if } \data( s, 0 ) \ge y_1 \\
        ( s, \vu ) & \text{if } \data( s, 1 ) = 1 \\
        ( s, \vu' ) & \text{else, s.t. }
                      \forall i \in [N] \cdot
                      \id( \vu'_i ) = \id( \vu_i )
                      \land
                      i \ne x_1 \implies \\& \rec( \vu'_i ) = \data( \vu_i )
                      \land
                      i = x_1 \implies \rec( \vu'_i ) = ( y_1 )
    \end{cases}
\end{align*}
\begin{align*}
    g_3( ( s, \vu ), ( x_1 ) )
    &=
    \begin{cases}
        ( s, \vu ) & \text{if } \data( s, 2 ) \ne 1 \\
        ( s, \vu ) & \text{if } \role( s, 0 ) = x_1 \\
        ( s, \vu ) & \text{if } \data( s, 0 ) \ne \rec( \vu_{x_1} )_0 \\
        ( s, \vu' ) & \text{else, s.t. }
                      \forall i \in [N] \cdot
                      \id( \vu'_i ) = \id( \vu_i )
                      \land
                      i \ne x_1 \implies \\& \rec( \vu'_i ) = \data( \vu_i )
                      \land
                      i = x_1 \implies \rec( \vu'_i ) = \vec{0}
    \end{cases} \\
    g_4( ( s, \vu ), ( x_1 ) )
    &=
    \begin{cases}
        ( s, \vu ) & \text{if } \data( s, 2 ) \ne 1 \\
        ( s, \vu ) & \text{if } \role( s, 0 ) \ne x_1 \\
        ( s', \vu ) & \text{else, s.t. }
                      \role( s', 0 ) = \role( s, 0 ), \\&
                      \data( s', 0 ) = \data( s, 0 ),\;
                      \data( s', 1 ) = 1, \text{ and } \\&
                      \data( s', 2 ) = \data( s, 2 )
    \end{cases}
\end{align*}

Recall from \cref{Sec:SyntaxSemantics} that a reverted transaction (such as a
failed require statement) is treated as a no-op. In $f$, these no-ops correspond
to the cases that send $( s, \vu )$ to $( s, \vu )$. Most of these cases, such
as the second case of $g_2$, are derived directly from the source code (line~\ref{line:req2}    
for this case). However, the first two cases for $f$, along with the first case
for each $g_i$, are special.

The first case for $f$ guards against transactions sent from the zero-account.
An important observation is that the sender, $x_1$, is compared to $\id( \vu_0
)$, rather than $0$. This is because $\cM$ is used to map all literal addresses
to users. For the first case of $f$, $\cM$ first maps $0$ to $\vu_0$, and then
$\id( \vu_0 )$ maps $\vu_0$ back to $0$\footnote{This indirection allows for
users to be readdressed in \cref{Def:Participation}}. Note that the second case
of $f$ is similar to the first case, in that the second case guards against
transactions from the smart contract account.

The first case of $g_1$ guards against calls to \texttt{constructor} after the
constructor has already the called. Similarly, the first case of $g_2$, $g_3$,
and $g_4$, each guard against calls to non-constructor functions before the
constructor has been called. These cases ensure that the contract is constructed
once and only once.

\section{MiniSol: Syntax and Semantics}
\label{Appendix:MiniSol}

\begin{figure}
  \scriptsize
  \input{diagrams/minisol_grammar}
  \normalsize
  \vspace{-1.2em}
  \caption{The formal grammar of the MiniSol language.}
  \vspace{-1.5em}
  \label{Fig:MiniSol}
\end{figure}

MiniSol is an extension of MicroSol that provides more communication primitives.
Specifically, MiniSol introduces: transactions with multiple calls; a currency
called \emph{Ether}; a built-in time primitive. The grammar for MiniSol is
presented in \cref{Fig:MiniSol}. As in MicroSol, the address type in MiniSol is
non-arithmetic, and mappings must relate addresses to integers. Note, however,
that MiniSol does allow for multi-dimensional mappings. Throughout the section,
we illustrate MiniSol  using the MiniSol program in \cref{Fig:Fund}.

\begin{figure}[t]
    \centering
\begin{lstlisting}[style=solidity,multicols=2]
contract Fund {
  mapping (address => uint) _funds;

  function () public payable { ¤\label{line:payable}¤
    _buy(); ¤\label{line:call-buy}¤
  }

  function _buy() internal { ¤\label{line:fund-buy}¤
    _funds[msg.sender] += msg.value; ¤\label{line:value}¤
  }

  function withdraw(address dst) public { ¤\label{line:fund-withdraw}¤
    uint amt = _balances[msg.sender];
    assert(amt <= balance(this)); ¤\label{line:balance}¤

    _balances[msg.sender] = 0;
    (dst).transfer(amt);
  }
}

contract FundManager {
  Fund _fund;

  mapping (address => bool) _owners;
  uint _end;
  address _dst;

  constructor(uint dur, address dst) public {
    _end = block.timestamp + dur;
    _dst = dst;
    _fund = new Fund();
    _owners[msg.sender] = true;
  }

  function add(address user) public {
    require(_owners[msg.sender]);
    _owners[user] = true;
  }

  function buy() public payable {
    require(_owners[msg.sender]);
    require(block.timestamp < _end); ¤\label{line:timestamp}¤
    fund.transfer(msg.value); ¤\label{line:transfer}¤
  }

  function withdraw() public {
    require(block.timestamp >= _end);
    _fund.withdraw.value(0)(_dst); ¤\label{line:call}¤
  }
}
\end{lstlisting}
    \caption{A MiniSol program that allows a \code{FundManager} to raise funds
             inside a \code{Fund}. The funds are forwarded to a predetermined
             destination after a duration.}
    \label{Fig:Fund}
\end{figure}

In MiniSol, a transaction consists of one \emph{or more} user calls. As in
MicroSol, the first call is performed by a user that is neither the zero-account
nor a smart contract account. Each subsequent call is made by the account of the
smart contract currently executing a functions (e.g., line~\ref{line:call} in
\cref{Fig:Fund}). As in most object-oriented languages, each function has a
\emph{visibility modifier} (see $\langle \text{Vis} \rangle$ in
\cref{Fig:MiniSol}). Functions with a \emph{public} visibility can be called by
users (e.g., line~\ref{line:fund-withdraw} in \cref{Fig:Fund}), whereas
functions with an \emph{internal} visibility can only be executed as part of an
ongoing call (e.g., the function at line~\ref{line:fund-buy} called at
line~\ref{line:call-buy} in \cref{Fig:Fund}). In either case, the inputs
\code{tx.origin} and \code{msg.sender} expose the addresses of the first and
most recent user to call a function during the current transaction,
respectively.

A currency called Ether is integrated into the MiniSol language. Each user has
a balance of Ether, that is exposed through the built-in \code{balance} function
(e.g., line~\ref{line:balance} in \cref{Fig:Fund}). A user that is neither the
zero-account nor a smart contract account can spontaneously acquire Ether
through \emph{minting}. A smart contract account can receive Ether through a
function marked \code{payable} (e.g., line~\ref{line:payable} in
\cref{Fig:Fund}). The amount of Ether sent to a payable function is accessible
through the \code{msg.value} input (e.g., line~\ref{line:value} in
\cref{Fig:Fund}). There are two ways that one smart contract can send Ether to
another smart contract. If a function is called by name, then the \code{.value}
modifier is used (e.g., line~\ref{line:call} in \cref{Fig:Fund}). Otherwise, a
built-in \code{transfer} function (e.g., line~\ref{line:transfer} in
\cref{Fig:Fund}) can be used to call a special unnamed function, called a
\emph{fallback function} (e.g., line~\ref{line:payable} in \cref{Fig:Fund}).

MiniSol also exposes a built-in time primitive to each function. In MiniSol,
time is measured with respect to the number of transactions that have been
executed. Transactions are executed in batches, known as \emph{blocks}. The
\code{block.number} input exposes the number of blocks executed so far. The
\code{block.timestamp} input exposes the real-world time at which the last block
was executed (e.g., line~\ref{line:timestamp} in \cref{Fig:Fund}). MiniSol
enforces that \code{block.timestamp} is non-decreasing.

The semantics of a MiniSol program are similar to those of a MicroSol program.
For a MiniSol program $\cC$ with transaction $\tran$, the transition function
$\bbrackfn{\cC}_{\cM}$ is still determined by the (usual) semantics of the
single transaction $\tran$. However, the state space and actions of each MiniSol
program also include information about the first caller of each transaction,
Ether, and time. The changes are outlined in \cref{Table:MiniSolSemantics}.

\begin{table*}[t]
  \small
  \centering
  \begin{tabular}{@{}lll@{}}
    \toprule Expression & \phantom{abc} & Interpretation \\

    \midrule
    $\role(s, i)$ && The value of the $i$-th address-typed state variable. \\
    $\data(s, 0)$ && The current value of \code{block.number}. \\
    $\data(s, 1)$ && The current value of \code{block.timestamp}. \\
    $\data(s, 2 + i)$ && The value of the $i$-th numeric-typed state variable. \\
    \midrule
    $\id( \vu_i )$ && The address of the $i$-th user. \\
    $\rec( \vu_i )_0$ && The balance of the $i$-th user. \\
    $\rec( \vu_i )_{1 + j}$ && The $j$-th mapping value of the $i$-th user. \\
    \midrule
    $\client(p, 0)$ && The value of \code{tx.origin}. \\
    $\client(p, 1)$ && The value of \code{msg.sender}. \\
    $\client(p, 2 + i)$ && The value of the $i$-th address-typed input variable. \\
    $\arg(p, 0)$ && The value of \code{msg.value}. \\
    $\arg(p, 1 + i)$ && The value of the $i$-th numeric-typed input variable. \\

    \bottomrule
  \end{tabular}
  \caption{The interpretation of a control state $s$, a user configuration $\vu$,
           and an action $p$ in MiniSol.}
  \label{Table:MiniSolSemantics}
\end{table*}

\section{Uniformity and Participation}
\label{Appendix:Uniformity}
In this section we show that MicroSol (and MiniSol) satisfy the uniformity
assumptions of PCMC given in \cref{Sec:Background} (i.e., there are finitely
many finite neighbourhoods). First, we show that if the set of implicit users is
finite, then each transaction of the network is performed against a
neighbourhood of bounded size. Second, we show that the users within the given
neighbourhood are interchangeable (up to address). Finally, we show that if all
transactions can be reduced to a finite set of addresses. From these results, it
follows that all MicroSol transactions are uniform.

\subsection{Reduction to a Bounded Neighbourhood}

Let $\cC$ be a contract. This section proves that if a PT
over-approximation for $\cC$ is finite, then every transaction of $\cC$ can be
computed by projecting the network onto a neighbourhood of bounded size. To do
this, a precise definition must first be given of an over-approximation and a
projection. The notion of over-approximation generalizes a PTG from
\cref{sec:communication}.

\begin{definition}[PT Over-Approximation]
    \label{Def:PTOverApprox}
    For $N \in \nat$ and $p \in \inputs( \cC, [N] )$, the tuple $( E, T, I )
    \subseteq \nat \times \nat \times \nat$ \emph{over-approximates} the PT,
    $\pt( \cC, N, p ) = ( \explicit, \\ \transient, \implicit )$, when $\{ i
    \mid \explicit( i, x ) \} \subseteq E$, $\{ i \mid \transient( i, x ) \}
    \subseteq T$, and $\implicit \subseteq I$. Furthermore, $( E, T, I )$ is a
    \emph{PT over-approximation} of $\cC$ when $( E, T, I )$ over-approximates
    all PT's of $\cC$.
\end{definition}

Note that when a PT Over-Approximation is finite, the over-approximation
corresponds to a \code{PTGBuilder} PTG. This is because each label in a
\code{PTGBuilder} PTG corresponds to precisely one user, with a specific
address. Therefore, the users that appear in a PT over-approximate also
over-approximate the local neighbourhood of a transaction. 

\newcommand{\view}{\textsf{view}}
\begin{definition}[View]
    Let $(E, T, I)$ be a PT over-approximation of $\cC$, $s \in \control( \cC,
    [N] )$, and $p \in \inputs( \cC, [N] )$. The \emph{$( E, T, I )$-view for
    $s$ and $p$} is: $$\view(E, T, I, s, p) \gets \{ \client(p, i) \mid i \in E
    \} \cup \{ \role(s, i) \mid i \in T \} \cup I$$
\end{definition}

In \cref{Sec:Sub:Localization} we use sets of addresses to represent local
neighbourhoods. Intuitively, each address corresponds to the user (or a
representative of the user) with the given address. In other words, the users in
a bundle are projected onto the neighbourhood, as defined below in
\cref{Def:Projection}.

\begin{definition}[Projection]
    \label{Def:Projection}
    Let $\cA_2 \subseteq \cA_1 \subseteq \nat$, $N = |\cA_1|$, $M = |\cA_2|$,
    and $\vu \in \user( \cC, \cA_1 )^N$. Then $\vcv \in \user( \cC, \cA_2)^M$ is
    a \emph{projection of $\vu$ onto $\cA_2$}, written $\vcv = \pi_{\cA_2}( \vu
    )$, if there exists a total, injective, order-preserving mapping, $\sigma
    [M] \rightarrow [N]$ such that $\cA_2 = \{ \id( \vcv_i ) \mid i \in [M] \}$
    and $\forall i \in [M] \cdot \vcv_i = \vu_{\sigma(i)}$.
\end{definition}

Notice that as $\sigma$ in \cref{Def:Projection} is order-preserving, each
projection is unique (with respect to $\cA$). This means that projections can be
compared directly. The following theorem shows that each transaction can be
reduced to a neighbourhood projection of finitely many clients.

\begin{theorem}
    \label{Thm:UniformityOne}
    Let $(E, T, I)$ be a PT over-approximation of $\cC$ with $N_0 = |E| + |T| +
    |I|$ finite. Then for all $N \ge N_0$ with $(S, P, f, s_0) = \lts(\cC, N)$,
    and for all $(s, u) \in S$, $p \in P$, if $(s', \vu') = f(s, \vu, p)$ and
    $\cA = \view(E, T, I, s, p)$, then $\forall i \in [N] \cdot \id(\vu_i) \not
    \in \cA \implies \vu_i = \vu_i'$ and $(s', \pi_{\cA}(\vu')) =
    \bbrackfn{\tran}_{N_0} \left( s', \pi_{\cA}(\vu), p \right)$.
\end{theorem}

\begin{proof}
    \cref{Thm:UniformityOne} draws two conclusions. (1) The first conclusion
    claims that $\forall i \in [N] \cdot \id(\vu_i) \not \in \cA \implies \vu_i
    = \vu_i'$. When $\id(\vu_i) \in \cA$, this holds trivially. Assume that
    $\id(\vu_i) \not \in \cA$. Then $\vu_i$ is not a participant. Then $f_p$ has
    no influence on $\vu_i$. By definition, $\vu_i = \vu_i'$, as desired. (2)
    The second conclusion claims that $(s', \pi_{\cA}(\vu')) =
    \bbrackfn{\tran}_{N_0} \left( s', \pi_{\cA}(\vu), p \right)$. By
    construction, all participants of $f_p$ are in $\cA$. Then $\pi_{\cA}(\vu)$
    contains all clients which influence $f_p$, while $\pi_{\cA}(\vu')$ contains
    all clients which are influenced by $f_p$. In addition, $\bbrackfn{\tran}$
    is the semantic interpretation of $\tran$. The source text of $\tran$ is
    defined independent from the size of the network, therefore
    $\bbrackfn{\tran}_N$ and $\bbrackfn{\tran}_{N_0}$ must perform the same
    sequence of synchronized and internal actions against the same set of
    clients. \qed
\end{proof}

\subsection{The Interchangeability of Users}

Let $\cC$ be a contract. This section shows that arbitrary users in $\cC$ are
interchangeable. We rely on the MicroSol restriction that address values are not
numeric.

\newcommand{\swap}{\textsf{swap}}
To discuss the interchangeability of users, we require a notion of swapping
users. Let $\cA$ be an address space, $\vu, \vu' \in \user( \cC, \cA )^N$, and
fix $x, y \in \cA$. Then $\vu'$ is an \emph{$(x, y)$ address swap} of $\vu$,
written $\swap( \vu, x, y )$, if $\forall i \in [N]$:
\begin{inparaenum}[(1)]
\item $\id(\vu_i) = y \implies \id(\vu'_i) = x$;
\item $\id(\vu_i) = x \implies \id(\vu'_i) = y$;
\item $\id(\vu_i) \not \in \{ x, y \} \implies \vu'_i = \vu_i$;
\item $\rec(\vu'_i) = \rec(\vu_i)$.
\end{inparaenum}
A similar notion is defined for the roles of a control process. Specifically,
let $s, s' \in \control( \cC, \cA )$. Then, $s'$ is an $(x, y)$ address swap of
$s$, if $\forall i \in [N]$:
\begin{inparaenum}[(1)]
\item $\role(s, i) = y \implies \role(s', i) = x$;
\item $\role(s, i) = x \implies \role(s', i) = y$;
\item $\role(s, i) \not \in \{ x, y \} \implies \role(s', i) = \role(s, i)$.
\end{inparaenum}
The extensions to $\inputs( \cC, \cA )$ and $\control( \cC, \cA ) \times \user(
\cC, \cA )^N$ are trivial.

\begin{lemma}
\label{Lemma:Interchangeable}
    Let $\cC$ be a contract, $N \in \nat$, $( S, P, f, s_0 ) = \lts(\cC, N)$,
    $p \in P$, and $\pt(\cC, N, p) = ( \explicit, \transient, \implicit )$. For
    every $s, t \in S$ and every $x, y \in \nat \setminus \implicit$, $f( \swap(
    t, x, y  ), \swap( p, x, y ) ) = \swap( f( s, p ), x, y)$.
\end{lemma}

\cref{Lemma:Interchangeable} says that address swaps commute with MicroSol
transactions. The lemma follows directly from the semantics of MicroSol, since
MicroSol restricts address comparisons to (dis)equality.

\cref{Lemma:Interchangeable} extends to address bijections. If $\cA_1, \cA_2
\subsetneq \nat \setminus \implicit$ are finite, transactions are also preserved
under any address bijection $\sigma: \cA_1 \rightarrow \cA_2$. Intuitively,
$\sigma$ is constructed using a sequence of swaps from $\cA_1$ to $\cA_2$. The
definition of $\sigma$ is inductive, with respect to the size of $\cA_1$.

\begin{theorem}
    \label{Thm:Permute}
    Let $\cC$ be a contract, $N \in \nat$, $( S, P, f, s_0 ) = \lts(\cC, N)$, $p
    \in P$, $\pt(\cC, N, p) = (\explicit, \transient, \implicit)$, and $\cA_1,
    \cA_2 \subseteq \nat \setminus \implicit$ be finite. If $\sigma: \cA_1
    \rightarrow \cA_2$ is bijective, then
    $f_{\sigma(p)}(\sigma(s)) = \sigma^{-1}(f_p(s))$.
\end{theorem}

\begin{proof}
    This follows by induction over the size of $\cA_1$. As $\sigma$ is a
    correspondence, then $|\cA_1| = |\cA_2|$. In the base case, $|\cA_1| = 1$,
    $\cA_1 = \{ a_1 \}$ and $\cA_2 = \{ a_2 \}$ with $a_1 \not \in \implicit$
    and $a_2 \not \in \implicit$. Then $\sigma(s) = \sigma^{-1}(s) = \swap(s,
    a_1, a_2)$ and $\sigma(p) = \sigma^{-1}(p) = \swap(p, a_1, a_2)$. Now
    assume that the inductive argument holds up to $k \ge 1$, and that
    $|\cA_1| = k + 1$. Then fix some $a_1 \in \cA_1$, and let $\cA_1' = \cA_1
    \setminus \{ a \}$ and $\cA_2' = \cA_1 \setminus \{ \sigma(a) \}$. Clearly
    $\tau: \cA_1' \rightarrow \cA_2'$, as defined by $\tau(x) = \sigma(x)$ is
    also a correspondence. Then $\sigma(s) = \swap(\tau(s), a, \sigma(a))$,
    $\sigma^{-1}(s) = \swap(\tau^{-1}(s), a, \sigma(a))$, $\sigma(p) =
    \swap(\tau(p), a, \sigma(a))$, and $\sigma^{-1}(p) = \swap(\tau^{-1}(p),
    a, \sigma(a))$. By the inductive hypothesis, $\tau$ and $\tau^{-1}$ are
    also reducible to sequences of swaps. \qed
\end{proof}

\subsection{Reduction to a Finite Address Set}

Let $\cC$ be an contract. This section proves that if the assumptions of
\cref{Thm:UniformityOne} holds for $\cC$, then by the results of
\cref{Thm:Permute}, the state space of each neighbourhood can be reduced to a
finite state space. Recall that all non-address state is already over a finite
domain.

\begin{theorem}
    \label{Thm:UniformityTwo}
    Let $(E, T, I)$ over-approximate PT's of $\cC$ with $N_0 = |E| + |T| + |I|$
    finite an let $I \subseteq \cA_0 \subseteq \cA^{*} \subseteq \nat$ with
    $|\cA_0| = N_0$. Then  for all $s \in \control( \cC, \cA^{*} )$, and for all
    $p \in \inputs( \cC, \cA^{*} )$, with $\cA = \view( E, T, I, s, p )$, there
    exists a correspondence $\sigma: \left( \cA \setminus I \right) \rightarrow
    \left( \cA_0 \setminus I \right)$ such that $f_p( s, \vu ) =
    \sigma^{-1}(f_{\sigma(p)}(\sigma(s), \sigma(\vu)))$.
\end{theorem}

\begin{proof}
    This proof follows in two parts. First, it is shown that $\sigma$ exists.
    This follows quite trivial. By construction, $|\cA| = N_0$, $|\cA_0| =
    N_0$, $I \subseteq \cA$, and $I \subseteq \cA_0$, therefore $|\cA
    \setminus I| = |\cA_0 \setminus I|$. Then a correspondence $\sigma: \cA
    \rightarrow \cA_0$ can be constructed arbitrarily. As $\sigma$ does not
    map to or from $I$, \cref{Thm:Permute} applies to $\sigma$. As $I$ is
    finite, \cref{Thm:UniformityOne} also applies. The conclusion of
    \cref{Thm:UniformityTwo} follows as a direct consequence. \qed
\end{proof}

From \cref{Thm:UniformityTwo}, every MicroSol transaction can be simulated using
a finite address space. By \cref{Thm:Permute}, this simulation also preserves
$k$-universal safety properties. This means
that there are finitely many distinguishable neighbourhoods, and that these
neighbourhoods are subsumed by $\cA$. Therefore, all MicroSol transactions are
uniform, as desired.

}{}

\end{document}